\documentclass[12pt, reqno]{amsart}
\usepackage{amsmath}
\usepackage{amssymb}
\usepackage{amsthm}

\usepackage{tabularx}
\usepackage{float}
\usepackage{color}

\usepackage[mathscr]{eucal}
\usepackage{mathrsfs}
\usepackage{psfrag}
\usepackage{fullpage}
\usepackage{color}
\usepackage{eucal}
\usepackage[dvips]{graphicx}
\usepackage{amsfonts}%
\providecommand{\U}[1]{\protect\rule{.1in}{.1in}}
\newtheorem{proposition}{Proposition}[section]
\newtheorem{theorem}[proposition]{Theorem}

\newtheorem{lemma}[proposition]{Lemma}

\newtheorem{remark}[proposition]{Remark}

\newtheorem{assumption}[proposition]{Assumption}

\newtheorem{con}[proposition]{Condition}

\numberwithin{equation}{section}
\numberwithin{proposition}{section}

\newcommand{\R}{\mathbb{R}}

\newcommand{\cC}{\mathcal{C}}
\newcommand{\Hh}{\mathcal{H}}

\newcommand{\N}{\mathbb{N}}

\newcommand{\nada}[1]{}
\newcommand{\imala}{ipMALA }
\newcommand{\diag}{\mathrm{diag}}

\numberwithin{equation}{section}
\numberwithin{proposition}{section}

\newcommand{\hf}{\hfill$\Box$}

\newcommand{\be}{\begin{equation}}
\newcommand{\ee}{\end{equation}}
\newcommand{\dd}{,{\dots},}
\newcommand{\norc}[1]{\| #1\|_{\mathcal{C}}}
\newcommand{\norcn}[1]{\| #1\|_{\mathcal{C}^N}}
\newcommand{\nor}[1]{\| #1\|}
\newcommand{\lanc}[2]{\langle #1, #2\rangle_{\mathcal{C}}}
\newcommand{\lancn}[2]{\langle #1, #2\rangle_{\mathcal{C}^N}}

\newcommand{\EE}{\mathbb{E}}
\newcommand{\lv}{\left\vert}
\newcommand{\rv}{\right\vert}
\newcommand{\ra}{\rightarrow}
\newcommand{\V}{\mathcal{V}}
\newcommand{\D}{\mathcal{D}}
\newcommand{\cN}{\mathcal{N}}
\newcommand{\ta}{\tilde{a}}
\newcommand{\tb}{\tilde{b}}

\newcommand{\tS}{\tilde{S}}

\newcommand{\bn}{^{(N)}}
\newcommand{\kn}{_k^N}
\newcommand{\Epin}{\mathbb{E}_{\pi^N}}
\newcommand{\less}{\lesssim}
\newcommand{\cH}{\mathcal{H}}
\newcommand{\cP}{\mathcal{P}}
\newcommand{\sub}{_{\mathbf{p}}}

\newcommand{\cCn}{(\mathcal{C}^N)}
\newcommand{\tSn}{(\tilde{S}^N)}
\newcommand{\id}{\mathrm{I}}

\begin{document}

\title{Optimal Scaling of the MALA algorithm with Irreversible Proposals for Gaussian targets}

\author{Michela Ottobre}
\address{Department of   Mathematics\\
Heriot Watt University,Edinburgh, EH14 4AS, Scotland }
 \email{ m.ottobre@hw.ac.uk}

\author{Natesh S. Pillai}
\address{Department of   Statistics\\
 Harvard University, Cambridge, MA, 02138-2901,USA }
 \email{ pillai@fas.harvard.edu}

\author{Konstantinos Spiliopoulos}
 \address{Department of  Mathematics and Statistics\\
 Boston University, Boston, MA, 02215, USA}
 \email{kspiliop@math.bu.edu}

\date{\today}

\begin{abstract}
It is well known in many settings that reversible Langevin diffusions in confining potentials converge to equilibrium exponentially fast.   Adding irreversible perturbations to the drift of  a Langevin diffusion that maintain the same invariant measure accelerates its convergence to stationarity.  Many existing works thus advocate the use of such non-reversible dynamics for sampling.
When implementing Markov Chain Monte Carlo algorithms (MCMC) using time discretisations of such Stochastic Differential Equations (SDEs), one can append the discretization with the usual Metropolis-Hastings accept-reject step and this is often done in practice because the accept--reject step eliminates bias. On the other hand, such a step  makes the resulting chain reversible. It is not known whether adding the accept-reject step preserves the faster mixing properties of the non-reversible dynamics. In this paper, we address this gap between theory and practice by analyzing the optimal scaling of MCMC algorithms constructed from proposal moves that are time-step Euler discretisations of an irreversible SDE, for high dimensional Gaussian target measures. We call the resulting algorithm the \imala, in comparison to the classical MALA algorithm (here {\em ip} is for irreversible proposal). In order to quantify how the cost of the algorithm scales with the dimension $N$, we prove invariance principles for the appropriately rescaled chain.   In contrast to the usual MALA algorithm, we show that there could be two regimes asymptotically: (i) a diffusive regime, as in the MALA algorithm and (ii) a ``fluid" regime where the limit is an ordinary differential equation. We provide concrete examples where the limit is a diffusion, as in the standard MALA, but with provably higher limiting acceptance probabilities. Numerical results are also given corroborating the theory.
 \end{abstract}

\maketitle


\section{Introduction}
In this paper, we analyze the scaling properties of high dimensional Markov Chain Monte Carlo (MCMC) algorithms constructed using non-reversible Langevin diffusions.
Consider a target measure
\be\label{tu}
\pi(dx)=\frac{1}{Z}e^{-{U(x)}}dx, \quad Z=\int_{\mathbb{R}^d} e^{-{U(x)}}dx.
\ee
It is known that, under mild assumptions on the potential $U(x)$,  the Langevin stochastic differential equation (SDE)
\begin{equation}\label{Eq:ReversibleLangevin}
dX_{t}=-\nabla U(X_{t})dt+\sqrt{2}dW_{t}
\end{equation}
has $\pi$ as its unique invariant measure and is $\pi$-ergodic. For appropriate test functions $f:\mathbb{R}^d \mapsto \mathbb{R}$, by the ergodic theorem, we have
\be \label{ergconv}
\frac{1}{t} \int_0^t f(X_s) ds \longrightarrow \int f(x) \pi(dx).
\ee
Thus the Langevin diffusion \eqref{tu} is a fundamental tool for sampling from the target measure $\pi$ or compute expectation of various functionals $\int f d\pi$, in view of \eqref{ergconv}. The Langevin SDE is time reversible -- its generator is a self-adjoint operator in $L_2(\pi)$. \par
The  drift of \eqref{Eq:ReversibleLangevin} can be modified without altering the invariant measure. Indeed, if  $\Gamma$ is a vector field such that  ${\rm div}( \Gamma e^{-U})=0$, then  diffusions of the form
\begin{equation}\label{Eq:IrreversibleLangevin}
dX_{t}=\left[-\nabla U(X_{t})+ \Gamma(X_{t})\right]dt+\sqrt{2}dW_{t},
\end{equation}
 also have $\pi$ as their invariant measure. The divergence free condition can be written as
\begin{equation} \label{eqn:divfree}
{\rm div} \Gamma =  \Gamma \nabla U.
\end{equation}
Observe that, for $\Gamma \neq 0$, the diffusions in \eqref{Eq:IrreversibleLangevin} are non-reversible. \par
\subsection{The problem}
It is known that adding a non-reversible component in a Langevin SDE could accelerate its convergence to stationarity. Indications of this phenomenon are the main results of \cite{HwangMaSheu2005,ReyBelletSpiliopoulos2014a,ReyBelletSpiliopoulos2014b}. The main result in \cite{HwangMaSheu2005} states that, among the family of non-reversible diffusions \eqref{Eq:IrreversibleLangevin}, the one with the smallest spectral gap corresponds to $\Gamma = 0$. Similar results were established from an asymptotic variance and large deviations point of view in \cite{ReyBelletSpiliopoulos2014a,ReyBelletSpiliopoulos2014b}. Broadly speaking, it is a well documented principle that non-reversible dynamics have better ergodic properties than their reversible counterparts. This observation has sparked a significant amount of research work in recent years and several papers have advocated the use of  non-reversible diffusions for sampling.

Numerical discretisations of \eqref{Eq:ReversibleLangevin} or \eqref{Eq:IrreversibleLangevin} do not necessarily inherit the ergodic properties of the continuous-time dynamics \cite{RT96}. In particular, discretised processes may not converge at all, or will have an invariant measure different from $\pi$.  To circumvent these  issues, in MCMC algorithms, practitioners often perform an additional Metropolis-Hastings  accept-reject step for proposals constructed from time discretisations of Langevin diffusions. For instance, the standard MALA algorithm is obtained when the discretised SDE is \eqref{Eq:ReversibleLangevin}  (see \textit{e.g.},  \cite{RobertsRosenthal1998} and \cite{Rossky}).
To the best of our knowledge, not much is known about whether making the chain reversible by adding the accept reject step preserves the faster mixing enjoyed by the non-reversible dynamics. In this article, we seek to address this important gap between theory and practice.

\subsection{Previous Work}
Many different approaches were pursued in recent works in order to exploit and analyse the beneficial effects of irreversibility in the algorithmic practice of MCMC methodology.
In particular: i)  Irreversible algorithms have been proposed and analyzed in
\cite{bierkens2015non,BouchardVollmerDoucet2016,diac:etal:2000, Ma, Krauth, Poncet} and references therein (on the matter see also \cite{Monmar,Ottirr} ); ii) algorithms that are obtained by discretising irreversible Markov processes in a way that the resulting Markov chain is still irreversible are studied in \cite{hor:91}, \cite{OPPS}; iii)  numerical algorithms that take advantage of the splitting of reversible-irreversible part of the equation are analyzed in \cite{LuSpiliopoulos2016, DuncanPavliotisZygalakis2016}. In addition, comparisons of MALA and of Langevin samplers without the accept-reject step have been performed in  \cite{DuncanLelievrePavliotis2016,DuncanPavliotisZygalakis2016}. In many cases of interest it has been observed that irreversible Langevin samplers have smaller mean square error when compared to the MALA algorithm, see \cite{DuncanPavliotisZygalakis2016}. The latter fact is related to the consideration that the variance reduction achieved by the irreversible Langevin sampler  can be  more significant than the the error due to the bias of the irreversible Langevin sampler.

\subsection{Our Contribution}
In this paper we take a different standpoint  and analyse the exact performance of a  Metropolis-Hastings algorithm, on certain Gaussian target densities, where the proposal is based on discretising (\ref{Eq:IrreversibleLangevin}) as opposed to (\ref{Eq:ReversibleLangevin}). We call this algorithm the {\em irreversible proposal} MALA or \imala\, for short. To quantify the efficiency of the \imala\, algorithm in high dimensions and compare it to the usual MALA algorithm, we study its optimal scaling properties and its limiting optimal acceptance probability \cite{RobertsRosenthal1998,RobertsRosenthal2001,PST}.
Optimal scaling aims to find the ``optimal size" of the local proposal variance as a function of the dimension.
The optimality criteria varies for different algorithms, but a natural choice  for algorithms that have a diffusion limit is the expected square jumping distance \cite{RobertsRosenthal2001}.

The basic mechanism for Metropolis-Hastings algorithms consists of employing a proposal transition density $q(x, y)$ in order to produce a reversible chain $\{x_k\}_{k=0}^{\infty}$ which has the target measure $\pi$ as invariant distribution. At step $k$ of the chain, a proposal move $y_{k+1}$ is generated by using $q(x,y)$, \textit{i.e.},
$y_{k+1} \sim q(x_k, \cdot)$. Then such a move is accepted with probability $\alpha(x_k, y_{k+1})$, where
\be\label{alpha1}
\alpha(x_k, y_{k+1})= \min\left\{1, \frac{\pi(y_{k+1}) q(y_{k+1},x_k)}{\pi(x_k) q(x_k,y_{k+1})}  \right\}\,.
\ee
The MALA algorithm is a Metropolis-Hastings algorithm with proposal move generated by a time-step discretisation of the Langevin equation \eqref{Eq:ReversibleLangevin}. The \imala\, algorithm we wish to analyze obtains proposals by discretising \eqref{Eq:IrreversibleLangevin}.  As explained before, any non-trivial $\Gamma$ satisfying the divergence free condition \eqref{eqn:divfree}
will preserve the invariant measure. A convenient choice is to pick $\Gamma$ such that
\be\label{meanconddd}
{\rm div} \Gamma =0 \,  \quad {\rm and} \quad  \Gamma \nabla U =0 \,.
\ee
A standard choice of $\Gamma(x)$ is
\be\label{antsymdr}
\Gamma(x)=S\nabla U(x),
\ee
 where $S$ is any antisymmetric matrix. A  more
elaborate discussion on other possible choices of $\Gamma(x)$ can be found in \cite{ReyBelletSpiliopoulos2014a}. The meaning of the conditions \eqref{meanconddd}
is straightforward: the flow generated by $\Gamma$ must preserve Lebesgue measure since it  is divergence-free; moreover,   the
micro-canonical measure  on the surfaces $\{U = z\}$ is preserved as well.

We make one further important assumption. For the rest of the paper,  we focus exclusively on Gaussian target measures. We believe most of our analysis should carry over to the important case in which the target measure has a Radon-Nikodym derivative with respect to a Gaussian measure using the methods in \cite{MattinglyPillaiStuart2011,PST}.

The main result of the paper is as follows. We consider Gaussian target measures  $\pi^N \sim \cN(0, \cC^N)$ on $\R^N$, where  $\cC^N:= diag\{\lambda_1^2 , ... , \lambda_N^2\}$ (see Section \ref{sec:2} for more details). Such a measure is clearly of the form \eqref{tu} for a quadratic potential $U$.
Therefore, in this case,  the general form of a Euler-discretisation of \eqref{Eq:IrreversibleLangevin} is given by
\begin{equation*}
y^{N}_{k+1}= x^{N}_k- \frac{\sigma_N^2}{2} x^{N}_{k} + \sigma_N^{\alpha}
\cC^N S^N x^{N}_{k}+ \sigma_N (\cC^N)^{1/2}z_{k+1}^N,
\end{equation*}
where $z^{N}_{k+1}\sim \cN(0, \mathrm{I}_N)$,  $\sigma_N=\ell/ N^{\gamma}$ and $\alpha>0$. The notation $S^N$ is used to stress that here $S^N$ is an $N \times N$ antisymmetric matrix.  The quantity $y^{N}_{k+1}$ is the proposal, and then a Markov chain is formed using the usual accept-reject mechanism (more details about the algorithm can be found in Section \ref{sec:3}). Let $x_k^N$ be the resulting \imala\, Markov chain. Interestingly, depending on whether $\alpha$ is bigger, equal or smaller than two, we will have different limiting behaviours. Broadly speaking, we show the following:
\begin{itemize}
\item if $\gamma< 1/6$ then the acceptance probability degenerates to zero exponentially quickly, and thus this case is not of practical interest.
\item if  $\alpha\geq 2$ then the optimal value of $\gamma$ is $\gamma=1/6$. We prove that the continuous interpolant of two subsequent steps of the MALA algorithm, see (\ref{continter}) for proper definition, has a diffusion limit.
     The cost of the algorithm is still $N^{1/3}$ like in the standard MALA case. However, if $S^N$ and $\alpha$ are chosen appropriately  (see specific examples in Sections \ref{S:NumericalApproxAccProb} and \ref{Eq:SimulationStudies}),  the limiting acceptance probability is higher. That is, the \imala will accept moves more frequently than MALA.
\item if $1\leq\alpha <2$ then take $\gamma\geq 1/6$. In this regime we  show that the cost of the algorithm is of the order $N^{\alpha\gamma}$, and one can choose $\alpha\gamma<1/3$. In addition, we  show that the continuous interpolant, see (\ref{continter}),  converges weakly to the solution of a deterministic ODE. It is also interesting  that the ODE that we get in the limit in this regime, can be related to Hamilton's ODE in HMC (see Theorem \ref{mainthm}) and we plan to investigate this in the future.
\end{itemize}

We would like to stress that in this paper we only study different scaling regimes when the algorithm is started {\em in stationarity}. In the case of the MALA algorithm it is a known fact that the optimal scaling out of stationarity differs from the optimal scaling in stationarity (the former being $O(N^{1/2})$ and the latter being $O(N^{1/3})$, see \cite{Christensenetal,jl1,jl2,kotmala}). It would be  relevant to address the same issue for the \imala  algorithm (especially in view of the fact that non-reversibility is known to speed up convergence,  although we do realise that the measure of efficiency we use here may well be unrelated to convergence rate) but we do not do it in this paper, as this would involve substantial further analysis.

The goal of this paper is to explore in a rigorous mathematical way what happens when a Metropolis-Hastings accept-reject step is applied to a proposal coming from discretisation of an irreversible Langevin diffusion. We find that there are different possible regimes and we characterize explicitly what is the limit of the continuous interplant of the chain as the number of steps goes to $\infty$, see Theorem \ref{mainthm}. We find that the irreversible perturbation matrix has non-trivial effects on the limiting dynamics (see Theorem \ref{mainthm}) and on the limiting average optimal acceptance probability, see (\ref{Eq:OptimalAcceptProb}) and Section \ref{S:NumericalApproxAccProb}. In terms of applications, the conclusion is that even though the introduction of the accept-reject step potentially offsets some of the advantages of irreversible perturbations \cite{HwangMaSheu2005,ReyBelletSpiliopoulos2014a,ReyBelletSpiliopoulos2014b}, some advantages may appear in certain cases. In particular, by appropriately choosing $S$ and $\alpha$ (e.g. $S=S_{1}$ and $\alpha>4$ as in Section  \ref{S:NumericalApproxAccProb}), we can obtain a diffusion limit but with limiting acceptance probabilities higher than that of MALA. In addition, even in the fluid regime, where the limit of the algorithm is an ODE, the acceptance probability can be the same as that of MALA, but the algorithm takes, in stationarity, $N^{\alpha\gamma}$ steps to explore the state space with $\alpha\gamma<1/3$ as opposed to $N^{1/3}$ steps for standard MALA.

Lastly, it is possible that for multi-modal targets the use of irreversible proposals may be beneficial as this is the case in the absence of the accept-reject step, see \cite{ReyBelletSpiliopoulos2014b};  an analysis of this fact is though beyond the scope of this paper.
The present paper certainly invites further research on this topic in different directions, including multi-modal targets, out of stationarity analysis and computational considerations of implementation issues.

The paper is organized as follows. In  Section \ref{sec:2} we introduce the notation used throughout the paper. In Section \ref{sec:3} we describe and motivate the algorithm that we will examine. Section \ref{sec:4}  contains the rigorous statement of our main result and its main implications. In Section \ref{sec:5} we give a detailed heuristic argument to explain how the main result is obtained and why it should hold. In Section \ref{S:NumericalApproxAccProb} we present some examples of potential choices for the antisymmetric matrix and the computation of the corresponding limiting acceptance probabilities. Section \ref{Eq:SimulationStudies} contains extensive simulation studies that demonstrate the theoretical results.  Rigorous proofs  are relegated to the Appendix. 

\section{Preliminaries and Notation}\label{sec:2}

Let $\left( \Hh, \langle\cdot, \cdot \rangle, \|\cdot\|\right)$
denote an infinite dimensional separable Hilbert space  with
the canonical norm derived from the inner-product.
Let $\cC$ be a  positive, trace class operator on $\Hh$
and $\{\phi_j,\lambda^2_j\}_{j \geq 1}$ be the eigenfunctions
and eigenvalues of $\cC$ respectively, so that
\be\label{for1}
\cC \varphi_j = \lambda_j^2 \varphi_j, \qquad \sum_{j} \lambda_j^2 < \infty.
\ee
The eigenvalues $\lambda_j^2$ are non-decreasing.
We assume a normalization under which $\{\phi_j\}_{j \geq 1}$
forms a complete orthonormal basis in $\Hh$.
Let $\pi$ be a Gaussian probability measure on $\Hh$ with covariance operator $\cC$, that is,
\begin{equation*}
\pi \sim \cN(0, \cC).
\end{equation*}
If $X^N$ is the finite dimensional space
\begin{equation*}
\Hh \supset X^N:=\textrm{span}\{\varphi_j\}_{j=1}^N
 \end{equation*}
spanned by the first $N$ eigenvectors of the covariance operator (notice that the space $X^N$ is isomorphic to $\R^N$), then for any fixed $N \in \N$ and $x \in \cH$ we denote by $\mathcal{P}^N(x)$ the projection of $x$ on $X^N$.
Moreover, the finite dimensional projection
 $\pi^N$ of the measure $\pi$ on $X^N$ is given by the  Gaussian measure
$$
\pi^N \sim \cN(0, \cC^N)
$$
where $\cC^N:= \diag\{\lambda_1^2 \dd \lambda_N^2\}$ (or, more precisely, $\cC^N:=\cP^N \circ \cC \circ \cP^N $).
Thus we have,
\begin{equation}\label{prodgas}
\pi^N:=\pi(x^N)= \pi(x^{1,N}\dd x^{N,N} ) = \frac{1}{(\sqrt{2 \pi})^N} \prod_{i=1}^N \frac{1}{\lambda_i}e^{- \frac{\lv x^{i,N}\rv^2}{2\lambda_i^2}}.
\end{equation}

Let  $\tS: \cH \ra \cH$ be a bounded linear operator. Thus there exists a constant $\kappa>0$ such that
\be\label{contCS}
\| \tS x\| \leq \kappa \|x\|, \qquad x \in \cH .
\ee
For any $N \in \N$ we can consider the projected operator $\tS^N$,   defined as follows:
$$
\tS^N x:= (\cP^N \circ \tS \circ \cP^N)x.
$$
The operator $\tS^N$ is also bounded on $\cH$. Since $X^N$ is isomorphic to $\R^N$, $\tS^N$ can be represented by an $N\times N$ matrix.  Throughout the paper we require the following:  $\tS$ is such that for any $N \in \N$, the matrix $\tS^N$ can be expressed  as the product of a symmetric matrix, namely $\cC^N$, and an antisymmetric matrix, $S^N$:
\be\label{seq}
\tS^N= \cC^N \, S^N, \quad \mbox{for every } N \in \N.
\ee

Throughout the paper we will use the following notation:
\begin{itemize}
 \item $x$ and $y$ are elements of the Hilbert space $\Hh$;
\item the letter $N$ is reserved to denote the  dimensionality of the space $X^N$ where
the target measure $\pi^N$ is supported;
\item $x^N$ is an element of $X^N$ { $\cong \R^N$} (similarly for {$y^N$ and the noise $\xi^N$}); the $j$-th component of the $N$-dimensional vector $x^N$ (in the basis $\{\varphi_j\}_{j=1}^N$)  is denoted by $x^{j,N}$.
\end{itemize}

We analyse Markov chains evolving in $\R^N$. Because the dimensionality of the space in which the chain evolves will be a key fact, we want to keep track of both the dimension $N$ and the step $k$ of the chain. Therefore,
\begin{itemize}
\item $x^N_k$ will denote the $k$-th step of the chain $\{x_k^N\}_{k \in \N}$ evolving in $\R^N$;
\item compatibly with the notation set above, $x^{i,N}_k$ is the $i$-th component of the vector $x^N_k \in \R^N$; \textit{i.e.}, $x^{i,N}_k= \langle x^N_k, \varphi_i\rangle$;
\item two (double) sequences of real numbers $\{A^{N}_k\}$ and $\{B^{N}_k\}$ satisfy $A^{N}_k \lesssim B^{N}_k$
if there exists a constant $K>0$ (independent of $N$ and $k$) such that
$$A^{N}_k\leq K B^{N}_k,$$
for all $N$ and $k$ such that $\{A^{N}_k\}$ and $\{B^{N}_k\}$ are defined. The notation $\lesssim$ is used for functions and random functions as well, with the constant $K$ being independent of the argument of the function and of the randomness.
\end{itemize}
As we have already mentioned,  $\|\cdot \|$ is the norm on $\Hh$, namely
$$
\|x\|^2:=\sum_{i=1}^{\infty}\langle x, \varphi_i\rangle^2.
$$
 With abuse of notation, we will also write
$$
\|x^N\|^2= \sum_{i=1}^N \lv x^{i,N}\rv^2 \,.
$$
We will also use the weighted norm $\|\cdot \|_{\cC}$, defined as follows:
$$
\| x\|^2_{\cC}:= \langle x, \cC^{-1} x \rangle = \sum_{i=1}^{\infty} \frac{\lv \langle x, \varphi_i\rangle \rv^2}{\lambda_i^2},
$$
for all $x \in \Hh$ such that the above series is convergent;  analogously,
$$
\| x^N\|^2_{\cC^N}:= \sum_{i=1}^N \frac{\lv x^{i,N} \rv^2}{\lambda_i^2} \,.
$$
If $A^N$ is an $N \times N$ matrix, $A^N= \{A^N_{i,j}\}$, we define
\begin{align}
\V_A^N &:= \sum_{i=1}^N {\lambda_i^2} \sum_{j=1}^N \lv A_{i,j}^N\rv^2
 \lambda_j^2 \nonumber \\
& = \EE_{\pi^N} \|(\cC^N)^{1/2} A^Nx^N\|^2 \,
= \EE_{\pi^N} \lv \langle (\cC^N)^{1/2}z^N, A^N x^N \rangle \rv^2 \label{asvar}
\end{align}
In the above $z^N \sim \cN(0, \id_N)$, $x^N \sim \pi^N$ and  $\Epin$ denotes expectation with respect to all the sources of noise contained in the integrand. We will often also write $\mathbb{E}_k$ and $\mathbb{E}_x$, to mean
$$
\mathbb{E}_k = \mathbb{E} [\cdot \vert x_k], \quad \mathbb{E}_x= \mathbb{E}[\cdot \vert x_k=x].
$$
Finally, many of the objects of interest in this paper depend on the matrix $S^N$ and on the parameters $\alpha$ and $\gamma$. When we want to stress such a dependence, we will add a subscript ${\sub}$, see for example the notation for the drift $d\sub $ in Theorem \ref{mainthm}.

\section{The \imala \, algorithm}\label{sec:3}
As we have already mentioned, the classical MALA algorithm is a Metropolis-Hastings algorithm which employs a proposal that results from a one-step discretisation of the Langevin dynamics \eqref{Eq:ReversibleLangevin}.
This is motivated by the fact that the dynamics \eqref{Eq:ReversibleLangevin} is ergodic   with unique invariant measure given by $\pi$ \eqref{tu} and can therefore be used to sample from $\pi$. The Langevin dynamics that samples from our target of interest,  the Gaussian measure $\pi^N$ in \eqref{prodgas}, reads as follows:
$$
dX_t = -(\cC^N)^{-1}X_t \,dt +  \sqrt{2} dW_t, \qquad X_t \in X^N,
$$
where $W_{t}$ is a N-dimensional standard Wiener process. If $\Delta$ is any positive definite, $N$-dimensional symmetric matrix, then  the equation
$$
dX_t = - \Delta (\cC^N)^{-1}X_t \,dt +  \sqrt{2\Delta} dW_t
$$
is still ergodic with invariant invariant measure $\pi^{N}$ in \eqref{prodgas}. Now notice that if $\cC$ is a trace class operator, then $\cC^{-1}$ is unbounded, so in the limit as $N\ra \infty$ any of the above two dynamics (and their discretisations) would lead to numerical instabilities. To avoid the appearance of unbounded operators, we can choose $\Delta=\cC^N/2$ and therefore consider the SDE
$$
dX_t = - \frac{1}{2} X_t \,dt+  (\cC^N)^{1/2} dW_t.
$$
Discretizing the above and using such a discretisation as a Metropolis-Hastings proposal would result in a well-posed MALA algorithm to sample from \eqref{prodgas}. However, as explained in the Introduction, here we want to analyze the MALA algorithm with irreversible proposal. As in \eqref{antsymdr}, we next consider the non-reversible SDE:
$$
dX_t = \left(- \frac{1}{2} X_t + S^N (\cC^N)^{-1}X_t \right) \,dt +(\cC^N)^{1/2} dW_t,
$$
where $S^N$ is any $N\times N$ antisymmetric matrix.
Again to avoid the appearance of unbounded operators, we modify the irreversible part of drift term and, finally, obtain the dynamics
\be\label{cdtbd}
dX_t = \left(- \frac{1}{2} X_t + \cC^N S^N X_t \right) \,dt +(\cC^N)^{1/2} dW_t.
\ee
Notice that, for any $x^N \in X^N$,
\be\label{stillinvar}
\nabla \cdot \left( (\cC^N S^N x^N) \, \pi^N \right)= \mathrm{Trace}(\cC^N S^N) \pi^N - \langle\cC^N S^N x^N,(\cC^N)^{-1}x^N \rangle \pi^N =0,
\ee
having used the antisymmetry of $S^N$ and the symmetry of $\cC^N$. Therefore, $\pi^N$ is invariant for the dynamics \eqref{cdtbd}. This justifies  using the Metropolis-Hastings proposal \eqref{propprodgas} below, which is a one-step Euler discretisation of the SDE \eqref{cdtbd}.

We now describe the \imala\, algorithm. If the chain is in $x_k^N$ at step $k$, the \imala\, algorithm has the proposal move:
\be\label{propprodgas}
y_{k+1}^N= x_k^N- \frac{\sigma_N^2}{2} x_k^N+ \sigma_N^{\alpha}
\cC^N S^N x_k^N+ \sigma_N (\cC^N)^{1/2}z_{k+1}^N
\ee
where
\be\label{defsigma}
\sigma_N = \frac{\ell}{N^{\gamma}}, \qquad \ell, \gamma>0
\ee
and $\alpha>0$. The proposal \eqref{propprodgas} is a generalised discretisation of \eqref{cdtbd}. We may indeed choose $\alpha$ so that, asymptotically, relative to the reversible drift and diffusion, the non-reversible drift dominates ($1\leq \alpha < 2$), vanishes ($\alpha>2$) or is balanced ($\alpha=2$). This will result in different scaling limits.  The choice of the parameters $\alpha, \gamma$ and $\ell$ will be further  discussed  in Section \ref{sec:4} below.   $S^N$ can be any antisymmetric $N \times N$ matrix and $z^N_{k+1}=(z^{1,N}_{k+1} \dd z^{N,N}_{k+1})$ is a vector of i.i.d standard Gaussians.   With this proposal, the acceptance probability is
\be \label{defaccprob}
\beta^N(x^N,y^N)= 1 \wedge e^{Q^N(x^N, y^N)}
\ee
where
\begin{align*}
Q^N(x^N,y^N)&= - \frac{\sigma_N^2}{8} \left( \norcn{y^N}^2- \norcn{x^N}^2\right)+\frac{1}{2}
\sigma_N^{2\alpha-2}\left( \norcn{ \cC^N S^N x^N}^2- \norcn{\cC^N S^N y^N}^2\right)\\
&+ 2\sigma_N^{\alpha-2}\langle S^Ny^N, x^N \rangle.
\end{align*}
The above expression for $Q^N$ is obtained with straightforward calculations, after observing that the proposal kernel $q(x^N,y^N)$ implied by  \eqref{propprodgas} is such that
$$
q(x^N,y^N)\propto \exp\left\{  - \frac{1}{2 \sigma_N^2}\norcn{ y^N-x^N+ \frac{\sigma_N^2}{2} x^N- \sigma_N^{\alpha} \cC^N S^N x^N}^2  \right\}.
$$
If $\tilde{\beta}^N \sim \mathrm{Bernoulli}(\beta^N(x^N,y^N))$, then the  Metropolis-Hastings chain  $\{x_k^N\}_{\{k \in \N\}}$ resulting from using the proposal
\eqref{propprodgas} can be written as follows
\be\label{chain}
x_{k+1}^N= \tilde{\beta}^N y_{k+1}^N+ (1- \tilde{\beta}^N)x_k^N = x_k^N+ \tilde{\beta}^N (y_{k+1}^N - x_k^N).
\ee

\section{Main results and Implications}\label{sec:4}
To understand the behaviour of the chain \eqref{chain}, we start by decomposing  it into its drift and martingale part; that is, let $\zeta>0$ and write
\be\label{dr-martal2}
x_{k+1}^N=x_k^N+ \frac{1}{N^{\zeta\gamma}} d\sub^N(x_k^N)+ \frac{1}{N^{\zeta\gamma/2}}M_{k, \mathbf{p}}^N
\ee
where the approximate drift $d\sub^N$ is defined as
\be\label{apprdrift}
d^N\sub(x_k^N):=N^{\zeta\gamma}\EE_k(x_{k+1}^N - x_k^N)
\ee
and the approximate diffusion $M_{k,\mathbf{p}}^N$ is given by
\begin{equation*}
M_{k, \mathbf{p}}^N:= N^{\zeta\gamma/2}\left[x_{k+1}^N - x_k^N - \EE_k(x_{k+1}^N - x_k^N) \right].
\end{equation*}
Using \eqref{propprodgas}-\eqref{defsigma} and \eqref{chain}, we  rewrite the approximate drift \eqref{apprdrift}  as follows:
\begin{align}
d^N\sub (x_k^N)&= N^{\zeta\gamma}\EE_k \left[ (1 \wedge e^{Q^N}) \left(- \frac{\ell^2}{2N^{2\gamma}}  x_k^N + \frac{\ell^{\alpha}}{N^{\alpha\gamma}} \cC^N S^N x_k^N\right) \right] \label{apprdrift1}\\
 &+ N^{\zeta \gamma}\EE_k \left[(1 \wedge e^{Q^N})\frac{\ell}{N^{\gamma}}(\cC^N)^{1/2}z_{k+1}^N \right] \,. \label{apprdrift2}
\end{align}
 Looking at \eqref{apprdrift1} it is clear that we need to examine two different cases, namely
\begin{description}
\item[i) Diffusive regime ]$2\gamma \leq \alpha \gamma$, \textit{i.e.}, $\alpha\geq2$
\item [ii) Fluid regime ]$2\gamma > \alpha \gamma$, \textit{i.e.}, $\alpha<2$.
\end{description}
The names of the above regimes will be clear after the statement of our main results,  Theorem \ref{thmspec} and Theorem \ref{mainthm}, see also Remark \ref{Rem:practimpl}.
We will show that, in order for the algorithm to have a well defined non-trivial limit, in the case i)  we should choose $\zeta=2$, whereas in case ii) we should choose $\zeta=\alpha$. For this reason, it is intended from now on that
\be\label{zeta=zetaalpha}
\zeta=\alpha \quad \mbox{if } \alpha <2  \qquad \mbox{and} \qquad \zeta =2 \quad \mbox{if } \alpha\geq 2.
\ee
With this observation in mind, we introduce  the continuous  interpolant of the chain $\{x_k^N\}$:
\be\label{continter}
x\bn (t)=(N^{\zeta\gamma}t-k)x_{k+1}^{N}+(k+1-N^{\zeta\gamma}t)x^{N}_k, \qquad t_k\leq t< t_{k+1},
\ee
where  $ t_k=k/N^{\zeta\gamma}$.

 We now come to state the main results of this paper, Theorem \ref{thmspec} and Theorem \ref{mainthm}.  Theorem \ref{mainthm} is the most general result. The assumptions under which Theorem \ref{mainthm} holds are a bit involved; we detail and motivate them in  Section \ref{SS:HeurAcceptProb}. In this section we first state Theorem \ref{thmspec}, which is just Theorem \ref{mainthm}, adapted to the case in which the sequence of matrices $S^N$ is chosen from a specific class. Namely,  we restrict our attention to the case in which each of the matrices  $S^{N}$ is in Jordan block  form:  we assume that   $S_{i,j}^{N}=J_{i}$ and $S_{j,i}^{N}=-J_{i}$ for $j=i+1$ and $i=1,3,5,7,\cdots, N-1$ and $S_{i,j}^{N}=0$ otherwise; here $\{J_i\}_i$ is an aribitrary sequence of real numbers.  To be more clear, when we refer to matrices $S^N$ in Jordan block form, we are referring to matrices of the form
\be\label{graphjordan}
\left( \begin{array}{cccccc}
0 & J_1 & 0 & 0 & 0 & 0 \\
-J_1 & 0 & 0 & 0 & 0 & 0\\
0 & 0 & 0 & J_2 &0 & 0\\
0 & 0 & -J_2 & 0 & 0 & 0 \\
0 & 0 & 0 & 0 & \ddots & \\
\end{array}
\right) \,.
\ee
 In this case, we can carry out explicit calculations,  construct examples  and demonstrate our results.

\begin{theorem}\label{thmspec}
Assume that the family of  anti-symmetric matrices $S^N$ is in Jordan block form (i.e. of the form \ref{graphjordan} as discussed above) and $x_0\sim \pi$.  Assume the constant
\be\label{finc1spec}
c_{1}=\lim_{N \rightarrow \infty} \frac{\Epin\norcn{\tS^N x^N}^2}{N^{2(\alpha-1)\gamma}} \quad \mbox { is finite}.
\ee
 Then, as $N \ra \infty$, the continuous interpolant $x\bn (t)$ of the chain  $\{x_k^N\}$ (defined in \eqref{continter}-\eqref{zeta=zetaalpha} and \eqref{chain}, respectively) started from the state $x_0^N= \cP^N(x_0)$, converges
weakly in $C([0,T];\mathcal{H})$ to the solution of the following equation
\be\label{generalSDE-ODE}
dx(t) = d\sub (x(t)) dt + D\sub \,  dW^{\cC}_t, \quad x(t) \in \cH,
\ee
where $W^{\cC}_t$ is a $\cH$-valued Brownian motion with covariance $\cC$. The drift coefficient $d\sub (x): \cH \ra \cH$ and the diffusion constant $D\sub\geq 0$ are defined as follows:

\be\label{drvc}
d\sub (x) : =\left\{
\begin{array}{l l }
- \frac{\ell^2}{2} h\sub^J x  & \mbox{if } \gamma=1/6 \mbox{ and } \alpha \geq 2\\
& \\
 0    &  \mbox{if } \gamma \geq 1/6 \mbox{ and } 1 \leq \alpha < 2 ,
\end{array}
\right.
\ee
 and
\be
D\sub : =\left\{
\begin{array}{l l}
\ell \sqrt{ h\sub^J} & \mbox{if } \gamma=1/6 \mbox{ and } \alpha \geq 2\\
 0  &  \mbox{if } \gamma \geq 1/6 \mbox{ and } 1 \leq \alpha < 2 \,.
\end{array}
\right.
\ee
The real constant $h\sub^J$ is defined as
\begin{align}\label{Eq:OptimalAcceptProbThm}
h\sub^J&:= 2\Phi\left(\frac{-\ell^{6}/32-a}{\sqrt{(\ell^{6}/16)+2a}}\right),
\end{align}
where $a=  2 \ell^{2(\alpha-1)}c_1$ and $\Phi$ is the standard Gaussian cumulative distribution function. Lastly, as measured by the number of steps, in stationarity, to explore the state space the cost of the algorithm in the diffusion case, $\alpha\geq 2$, is $N^{1/3}$ whereas in the fluid case, $1 \leq \alpha<2$, it is $N^{\alpha\gamma}$.
\end{theorem}

We now move on to stating Theorem \ref{mainthm}.  The precise statement of the assumptions of the theorem is deferred to  Section \ref{sec:5}. After stating Theorem \ref{mainthm},  we first compare it with Theorem \ref{thmspec}  (Remark \ref{rem:compare}) and then  make several observations regarding practical implications of the statement (Remark \ref{Rem:practimpl}). We stress that  the operator $\tilde{S}$ appearing in the statement of the theorem is as in Section \ref{sec:2} (see \eqref{contCS}-\eqref{seq}). In particular, the projected operator $\tilde{S}^N$ is, for every $N$, the product of the matrix $\cC^N$ and of a generic  antisymmetric matrix $S^N$,  appearing in the proposal \eqref{propprodgas}.

\begin{theorem}\label{mainthm}
Let Assumption \ref{AssCLT}, Assumption \ref{Ass2}, Condition \ref{Ass1} and Assumption \ref{extrassT2} hold and let $x_0\sim \pi$.  Then, as $N \ra \infty$, the continuous interpolant $x\bn (t)$ of the chain  $\{x_k^N\}$ (defined in \eqref{continter}-\eqref{zeta=zetaalpha} and \eqref{chain}, respectively) started from the state $x_0^N= \cP^N(x_0)$, converges
weakly in $C([0,T];\mathcal{H})$ to the solution of the following equation
\be\label{generalSDE-ODE}
dx(t) = d\sub (x(t)) dt + D\sub \,  dW^{\cC}_t, \quad x(t) \in \cH,
\ee
where $W^{\cC}_t$ is a $\cH$-valued Brownian motion with covariance $\cC$. In addition,   the drift coefficient $d\sub (x): \cH \ra \cH$ and the diffusion constant $D\sub\geq 0$ are defined as follows:
\be\label{drvc}
d\sub (x) : =\left\{
\begin{array}{l l }
- \frac{\ell^2}{2} h\sub x  & \mbox{if } \gamma=1/6 \mbox{ and } \alpha > 2\\
& \\
 - \frac{\ell^2}{2} h\sub x + \tau\sub \ell^{2} \tS x
& \mbox{if } \gamma=1/6 \mbox{ and } \alpha = 2 \\
& \\
 \tau\sub \ell^{\alpha} \tS x    &  \mbox{if } \gamma \geq 1/6 \mbox{ and } 1 \leq \alpha < 2 ,
\end{array}
\right.
\ee
 and
\be
D\sub : =\left\{
\begin{array}{l l}
\ell \sqrt{ h\sub} & \mbox{if } \gamma=1/6 \mbox{ and } \alpha \geq 2\\
 0  &  \mbox{if } \gamma \geq 1/6 \mbox{ and } 1 \leq \alpha < 2
\end{array}
\right.
\ee
The real constants $h\sub$ and $\tau\sub$ are defined in (\ref{hS}) and in (\ref{mhs1}) respectively.
\end{theorem}
\begin{proof} See Appendix \ref{AppendixA} and \ref{AppendixB}.
  \end{proof}

\begin{remark}\label{rem:compare}\textup{If the sequence of matrices $S^N$ is in Jordan block form, then all the assumptions of Theorem \ref{mainthm} are satisfied, provided \eqref{finc1spec} holds. \footnote{Indeed, in this case Assumption \ref{AssCLT}
and Assumption \ref{Ass2} presented in Section \ref{sec:5} are satisfied with $c_1$ as in \eqref{finc1spec} and $c_2=c_3=0$. This makes Assumption \ref{extrassT2} easy to verify. Moreover, Condition \ref{Ass1} is trivially satisfied  for matrices in Jordan block form. Detailed comments on this can be found in Section \ref{sec:5}, see comments after \eqref{Eq:OptimalAcceptProb}. } It should be clearly said that, while examples and calculations in the case of Theorem \ref{thmspec} are very explicit,  we have not been able to find examples of matrices $S^N$ such that the constant $\tau\sub$ in Theorem \ref{mainthm} is non-vanishing.}

\textup{However, it is important to stress that even when $\tau\sub=0$ it is still possible that the average acceptance probability (given, in the limit, by $h\sub$) is not tending to one and it is of order 1  instead, see the examples in Sections \ref{S:NumericalApproxAccProb} and \ref{Eq:SimulationStudies}. Notice that $\tau\sub =0$ when $b=2a$, with $b$ and $a$ are parameters depending on the choice of $S^N, \alpha$ and $\gamma$,  defined in \eqref{alb} and \eqref{albb}. Cfr Table \ref{tab:1} and Table \ref{tab:2} to have a summary of how $h\sub$ and $\tau\sub$ depend on $a$ and $b$. As the numerical studies of Sections \ref{S:NumericalApproxAccProb} and \ref{Eq:SimulationStudies} demonstrate, even in the setup of Theorems \ref{thmspec}, for which we can construct practical algorithms, we end up with interesting non-trivial results.
}
\end{remark}

We now  comment on the practical implications of Theorems \ref{thmspec} and \ref{mainthm}.
\begin{remark}\label{Rem:practimpl}
\textup{The  limiting drifts (and corresponding diffusion coefficients) appearing in \eqref{drvc} correspond to the two different regimes i) and ii), which we identified before the statement of the theorem.    The choice $\gamma=1/6$ in case (i)  (\textit{i.e.}, in the diffusive regime) and $\gamma \geq 1/6$ in the fluid regime will be fully motivated  in  Section \ref{SS:HeurAcceptProb}. The various results obtained in the different regimes are summarised in Tables \ref{tab:2} and \ref{tab:1} below.} 
\begin{itemize}
\item\textup{In the regime (i),  the effective time-step implied by the interpolation \eqref{zeta=zetaalpha}--\eqref{continter},
is $N^{-2\gamma}=N^{-1/3}$. Therefore, if $\gamma=1/6$ and $\alpha\geq 2$,  Theorem \ref{mainthm} implies that the optimal scaling for the proposal variance when the chain is in its  stationary regime is of the order $N^{-1/3}$. That is, the cost of the algorithm (in terms of number of steps needed, in stationarity, to explore the state space) is $O(N^{1/3})$. This is the same scaling as obtained in \cite{RobertsRosenthal1998, PST} for the MALA algorithm. Therefore, in the regime $(i)$, the \imala\,algorithm has the same scaling properties of the usual MALA. 
More discussions on this can be found in Section \ref{sechad}.
}
\item \textup{In the case  $\gamma=1/6$ and $\alpha> 2$,  one can construct specific matrices $S$, which result in considerably higher limiting acceptance probabilities than in the classical MALA case, see Sections \ref{S:NumericalApproxAccProb} and \ref{Eq:SimulationStudies}.
}
\item \textup{ When $\alpha<2$ and $\gamma\geq 1/6$, the (rescaled) chain converges to a fluid limit; \textit{i.e.}, the limiting behaviour of the chain is described by an ODE. Such an ODE still admits our target as invariant measure.
   With the same reasoning as above, the main result implies that, in this regime,  the cost  of the algorithm  is of the order $N^{\alpha\gamma}$ and one can choose $\alpha\gamma<1/3$. In the fluid regime, one may expect that the chain gets to get stuck in the limit as $N\rightarrow\infty$. However, we also comment here that in the simulations that we performed we did not see the chain getting stuck in this regime, which may imply that one may need to go to very high dimensions to observe such issues.
}
\end{itemize}
\end{remark}

\begin{table}[h!]
\begin{tabular}{|c|c|c|}
\hline
\multicolumn{3}{|c !{\vrule width 2pt}}{Setting of Theorem \ref{thmspec}: $c_{j}= 0$ for $j=2,3$ and $b= 2a$}\\
\hline
\multicolumn{3}{|c !{\vrule width 2pt}}{$\gamma=1/6 \Rightarrow $ acc. prob is $O(1)$}\\
\hline
$\alpha\geq 2$ & $dx_t =- \frac{\ell^2}{2}h\sub^J x_t dt+ \ell \sqrt{h^J\sub} dW_t^{\cC} $ & $N^{1/3}$ \\
\hline
$1 < \alpha <2$ & $dx_t = 0$ & $N^{\alpha\gamma}$\\
\hline \hline
\multicolumn{3}{|c !{\vrule width 2pt}}{$\gamma>1/6 \Rightarrow $ acc. prob tends to one}\\
\hline
$\alpha \geq 2$ & \multicolumn{2}{c !{\vrule width 2pt}}{too costly $N^{\alpha \gamma}(> N^{1/3})$ } \\
\hline
$1 \leq \alpha <2$ & $dx_t = 0$ & $N^{\alpha\gamma}$\\
\hline
\end{tabular}
\smallskip
\caption{The possible different regimes when the antisymmetric matrix $S$ takes the Jordan block form. The constants $c_j$ are those appearing in Assumption \ref{Ass2}; such constants also determine the value of the parameters $a$ and $b$ (defined after \eqref{asympQnormal1})}
\label{tab:2}
\end{table}

\begin{table}[h!]
\begin{tabular}{|c|c|c|}
\hline
\multicolumn{3}{|c !{\vrule width 2pt}}{Setting of Theorem \ref{mainthm}: case $c_{j}\neq 0$ for $j=1,2,3$ and $b\neq 2a$}\\
\hline
\multicolumn{3}{|c !{\vrule width 2pt}}{$\gamma=1/6 \Rightarrow $ acc. prob is $O(1)$}\\
\hline
$\alpha>2$ & $dx_t =- \frac{\ell^2}{2}h\sub x_t dt+ \ell \sqrt{h\sub} dW_t^{\cC} $ & $N^{1/3}$ \\
\hline
$\alpha=2$ & $dx_t =- \frac{\ell^2}{2}h\sub x_t dt+ \tau\sub \ell^2\tilde{S}x_t dt+ \ell \sqrt{h\sub} dW_t^{\cC} $ &
$N^{1/3}$ \\
\hline
$1 \leq \alpha <2$ & $dx_t = \tau\sub \ell^{\alpha} \tilde{S} x_t dt$ & $N^{\alpha\gamma}$ \\
\hline \hline
\multicolumn{3}{|c !{\vrule width 2pt}}{$\gamma>1/6 \Rightarrow $ acc. prob is $O(1)$}\\
\hline
$\alpha \geq 2$ & \multicolumn{2}{c !{\vrule width 2pt}}{too costly $N^{\alpha \gamma}(> N^{1/3})$ } \\
\hline
$1 \leq \alpha <2$ & $dx_t = \tau\sub \ell^{\alpha} \tilde{S} x_t dt$ & $N^{\alpha\gamma}$\\
\hline
\end{tabular}
\smallskip
\caption{The above table summarizes the general scenario described by Theorem \ref{mainthm} in the case in which  $c_j\neq 0$ for $j=1,2,3$ and $b\neq 2a$.   The constants $c_j$  are those appearing in Assumption \ref{Ass2}; such constants also determine the value of the parameters $a$ and $b$ (defined after \eqref{asympQnormal1}). Table \ref{tab:2} is a subcase of this table. Indeed, Theorem \ref{mainthm} also covers the following:  if $b=2a$ and $c_j\neq 0$ for at least one $j$ then all of the above holds with $\tau\sub=0$. If $b=2a$ and all the $c_j$'s vanish, then all of the above holds with $\tau\sub=0$.
}
\label{tab:1}
\end{table}

\section{Heuristic derivation of the diffusion limit}\label{sec:5}

In this section we give heuristic arguments to  explain how one can formally obtain the  diffusion limit  of Theorem \ref{mainthm} for the chain $\{x_k^N\}_k$.
We stress that the arguments of this section are only formal; therefore, we often use the notation $``\simeq"$, to mean ``approximately equal". We write  $A\simeq B$ when $A=B+$ ``terms that are negligible" as $N$ tends to infinity; we then rigorously justify these approximations,
and the resulting limit theorems, in the Appendix.

The goal of this section is to motivate (a): the proofs of Theorems \ref{thmspec} and \ref{mainthm} that are presented in the Appendix \ref{AppendixA} and \ref{AppendixB}, and (b): the Assumptions  \ref{AssCLT}-\ref{extrassT2} needed for the general limiting result to hold.

In particular, in Subsection \ref{SS:HeurAcceptProb}, we present the heuristic derivation and the intuition behind it for the asymptotic acceptance probability as the dimension increases.  On the way of doing so,  we present and justify the assumptions that need to be imposed in order for the appropriate scaling limits to hold. In addition, Remark \ref{remonassumptions} discusses at length these assumptions and presents examples where they are expected to hold. Subsections \ref{sechad} and \ref{SS:HeurDiffCoeff} present the heurustic calculations and motivation for the  drift and diffusion coefficient respectively of the limiting continuous time interpolation of the resulting Markov chain. The heuristic discussion of Subsections \ref{SS:HeurAcceptProb}, \ref{sechad} and \ref{SS:HeurDiffCoeff} is made rigorous in Appendix \ref{AppendixA} and \ref{AppendixB}.

\subsection{Study of the acceptance probability and statement of main assumptions}\label{SS:HeurAcceptProb}
In order to understand the behaviour of the chain, it is crucial to gain intuition about the acceptance probability $\beta^N$. While attempting to improve such an intuition in this section, we also present the motivation behind the assumptions needed for Theorem \ref{mainthm}.\par
 Using \eqref{propprodgas}, a more useful and detailed expression for $Q^N$ (which was introduced just after \eqref{defaccprob}) is
\begin{align}\label{QNdecomp}
Q^N(x^{N},y^N)&= \bar{Q}^N(x^N,y^N)+Q_{\alpha}^N (x^N,y^N)
\end{align}
where $\bar{Q}^N$ contains all the terms that come from the reversible part of the proposal while $Q^N_{\alpha}$ contains all the terms that come from the irreversible part of the proposal; namely,
\begin{align*}
\bar{Q}^N :=& \left(-\frac{\sigma_N^3}{4}+ \frac{\sigma_N^5}{8}\right)\langle x^N, (\cC^N)^{1/2}z^N \rangle_{\cC^N} + \frac{\sigma_N^4}{8} \left( \norcn{x^N}^2- \|z^N\|^2\right)\\
& -\frac{\sigma_N^6}{32} \norcn{x^N}^2
\end{align*}
and
\begin{align}
Q^N_{\alpha}:=& \frac{1}{2}\sigma_N^{2\alpha} \left( \norcn{\tilde{S}x^N}^2- \norcn{\tilde{S}^N (\cC^N)^{1/2}z^N}^2\right)- \left(2 \sigma_N^{\alpha-1}+ \frac{1}{4}\sigma_N^{3+\alpha} \right)\langle{(\cC^N)^{1/2}z^N},{S^Nx^N}\rangle \label{4.14}\\
& - \left(2 \sigma_N^{2\alpha-2}+ \frac{1}{4}\sigma_N^{2\alpha+2} \right)\norcn{\tilde{S}^Nx^N}^2
- \left( \sigma_N^{2\alpha-1}- \frac{1}{2}\sigma_N^{2\alpha+1} \right)\lancn{\tilde{S}^N (\cC^N)^{1/2}z^N}{\tilde{S}^Nx^N} \label{4.15}\\
& -  \sigma_N^{3\alpha -1} \lancn{\tilde{S}^N (\cC^N)^{1/2}z^N}{(\tilde{S}^N)^2x^N} - \frac{1}{2}\sigma_N^{4\alpha -2} \norcn{(\tilde{S}^N)^2 x^N}^2\\
& - \sigma_N^{3\alpha-2}\left( 1- \frac{\sigma_N^2}{2}\right)\lancn{\tilde{S}^Nx^N}{(\tilde{S}^N)^2 x^N}.\label{4.17}
\end{align}
For $N$ large, from \eqref{defsigma} we have $\sigma_N^{\alpha-1}>> \sigma_N^{3+\alpha}$; so we expect that, asymptotically, the last term in \eqref{4.14} will be negligible. The same reasoning can be applied to the terms in \eqref{4.15}. Therefore,    irrespective of the choice of $\gamma$ and $\alpha$, we have the approximation
\begin{align}
Q^N_{\alpha} \simeq & \frac{1}{2}\sigma_N^{2\alpha} \left( \norcn{\tS^N x^N}^2
- \norcn{\tS^N (\cC^N)^{1/2} z^N}^2\right)- 2 \sigma_N^{\alpha-1} \langle{(\cC^N)^{1/2}z^N},{S^Nx^N}\rangle\nonumber\\
& - 2 \sigma_N^{2\alpha-2}\norcn{\tS^N x^N}^2
- \sigma_N^{2\alpha-1}\lancn{ \tS^N (\cC^N)^{1/2}z^N}{\tS^N x^N}\nonumber\\
&- \sigma_N^{3\alpha -1} \lanc{\tS^N (\cC^N)^{1/2}z^N}{(\tS^N)^2 x^N} - \frac{1}{2}\sigma_N^{4\alpha -2} \norc{(\tS^N)^2 x^N}^2\nonumber\\
& - \sigma_N^{3\alpha-2} \left( 1- \frac{\sigma_N^2}{2}\right)\lancn{\tS^N x^N}{(\tS^N)^2 x^N}  \,. \label{zerot}
\end{align}
 We further observe that,  in stationarity (and again irrespective of the choice of $\gamma, \alpha$ and $S^N$) the term
$$
\frac{1}{2}\sigma_N^{2\alpha} \left( \norcn{\tS^N x^N}^2- \norcn{\tS^N (\cC^N)^{1/2} z^N}^2\right)
$$
is smaller than the term
$$
- 2 \sigma_N^{2\alpha-2}\norcn{\tS^N x^N}^2
$$
in the sense that it is  always of lower order in $N$. Moreover, due to the skew-symmetry of $S^N$, the term \eqref{zerot} is identically zero:
$$
\lancn{\tS^N x^N}{(\tS^N)^2 x^N}= \lancn{\cC^N S^N x^N}{\cC^N S^N \cC^N S^N x^N}=
\langle \cC^N S^N x^N,  S^N \cC^N S^N x^N\rangle =0\, .
$$
 Therefore we can make the further approximation
\begin{align}
Q^N_{\alpha} \simeq & - 2 \sigma_N^{\alpha-1} \langle{(\cC^N)^{1/2}z^N},{S^Nx^N}\rangle
- 2 \sigma_N^{2\alpha-2}\norcn{\tS^N x^N}^2 \label{feqan}\\
&
- \sigma_N^{2\alpha-1}\lancn{ \tS^N (\cC^N)^{1/2}z^N}{\tS^N x^N}
- \frac{1}{2}\sigma_N^{4\alpha -2} \norcn{(\tS^N)^2 x^N}^2\\
&- \sigma_N^{3\alpha -1} \lancn{\tS^N (\cC^N)^{1/2}z^N}{(\tS^N)^2 x^N}
=: R^N_{\alpha} \,. \label{leqan}
\end{align}
As a result of the above reasoning we then have the heuristic approximation
\be\label{qaapprra}
Q^N_{\alpha} \simeq R^N_{\alpha} \, .
\ee
We now heuristically argue that the only sensible choice for $\gamma$ is $\gamma=1/6$ if $\alpha \geq 2$ and $\gamma\geq 1/6$ if $\alpha<2$. In order to do so we look at the decomposition \eqref{QNdecomp} of $Q^N$ and recall the definition \eqref{defaccprob} of the acceptance probability;  we then write
$$
\beta^N= 1 \wedge (e^{\bar{Q}^N}e^{Q^N_{\alpha}}).
$$
 Let us start by looking at $\bar{Q}^N$:
if we start the chain in stationarity, i.e. $x_0^N\sim \pi^N$  then $x^{N}_k \sim \pi^N$ for every $k \geq 0$.  In particular, if   $x^N \sim \pi^N$ then $x^N$ can be represented as $x^N= \sum_{i=1}^N \lambda_i \rho_i \varphi_i$, where $\rho_i$ are i.i.d.~ $\cN(0,1)$. We can therefore use the law of large numbers and observe that
$$\|x^N\|_{\cC^N}^2=\sum_{i=1}^N \lv\rho_{i} \rv^2 \simeq N . $$
Similarly, by the Central Limit Theorem (CLT)  the term $\langle x^N, (\cC^N)^{1/2} z^N\rangle_{\cC^N}$ is $O(N^{1/2})$ and converges to a standard Gaussian as $N\ra \infty$. Again by the CLT we can see that also the term $\left( \norcn{x^N}^2- \|z^N\|^2\right)$ is $O(N^{1/2})$ in stationarity.  Therefore, recalling \eqref{defsigma}, one has
$$
\bar{Q}^N \simeq -\frac{\ell^3}{4N^{3\gamma}}
\langle x^N, (\cC^N)^{1/2} z^N\rangle_{\cC^N} -
\frac{\ell^6}{32 N^{6\gamma}} \norcn{x^N}^2\,.
$$
With these observations in place we can then argue the following:
\begin{itemize}
\item If $\gamma=1/6$ then $\bar{Q}^N$ is $O(1)$ (in particular, for large $N$ it converges to a Gaussian with finite mean and variance, see \eqref{QNapprgaussian});  therefore $e^{\bar{Q}^N}$ is $O(1)$ as well
\item If $\gamma <1/6$ then $\bar{Q}^N \ra -\infty$ (more precisely $\EE \bar{Q}^N \ra -\infty$), hence
$e^{\bar{Q}^N} \ra 0$
\item If $\gamma >1/6$ then $\bar{Q}^N \ra 0$, hence
$e^{\bar{Q}^N} \ra 1$
\end{itemize}
Let us now informally discuss the acceptance probability $\beta^N$  in each of the above three cases, taking into account the behaviour of $Q^N_{\alpha}$ as well. To this end it is worth noting the following:
\begin{align}
&\EE_{\pi^N} \norcn{\tS^N x^N}^2 = \EE_{\pi^N} \| (\cC^N)^{1/2}S^Nx^N\|^2=   \V_S^N= \Epin \lv \langle{\cC^{1/2}z^N},{S^N x^N}\rangle \rv^2 \label{constc1=}\\
& \EE_{\pi^N} \norcn{(\tS^N)^2 x^N}^2 = \V_{S\tS}= \Epin \lv \lancn{ \tS^N (\cC^N)^{1/2}z^N}{\tS^N x^N} \rv^2 , \label{constc2=}
\end{align}
where the first equality in \eqref{constc1=} holds because
$$
\| (\cC^N)^{1/2}S^Nx^N\|^2=  \norcn{\tS^N x^N}^2 \qquad \mbox{for all } x,
$$
and the others follow from \eqref{asvar}. By \eqref{constc1=}- \eqref{constc2=} and
 \eqref{feqan} - \eqref{leqan},  as $N \ra \infty$, the (average, given $x$, of the) quantity $Q^N_{\alpha}$ can only tend to zero, go to $- \infty$ or be $O(1)$, but it cannot diverge to $+\infty$. With this in mind,
\begin{description}
\item[a)]If $\gamma=1/6$ (that is, $\bar{Q}^N$ is $O(1)$) then one has the following two sub-cases
\begin{description}
\item[a1)] $\gamma=1/6$ and either ${Q^N_{\alpha}}\ra 0$ or $Q^N_{\alpha}$ is $O(1)$. In this case the limiting acceptance probability does not degenerate to zero (and does not tend to one)
\item[a2)] $\gamma=1/6$ and  ${Q^N_{\alpha}}\ra - \infty$. In this case $\beta^N \ra 0$.
\end{description}
\item[b)]If $\gamma <1/6$  then $e^{Q^N}$ can only tend to zero (if $Q^N_{\alpha}\ra 0$, to $-\infty$ or is $O(1)$). This is not a good regime as in this case   $\beta^N$ would overall tend to zero.
\item[c)]If $\gamma >1/6$  we have three sub-cases
\begin{description}
\item[c1)] $\gamma >1/6$ and ${Q^N_{\alpha}}$ is $O(1)$. In this case $\beta^N$ is $O(1)$.
\item[c2)] $\gamma >1/6$ and ${Q^N_{\alpha}} \ra 0$, in which case $\beta^N \ra 1$.
\item[c3)] $\gamma >1/6$ and ${Q^N_{\alpha}}\ra - \infty $ , in which case $\beta^N$ degenerates to zero.
\end{description}
\end{description}

We clearly want to rule out all the cases in which $\beta^N$ tends to zero,
(because this  implies that the chain  is getting stuck). Therefore, if we want to ``optimize" the behaviour of the chain,  the only  good choices are a1),  c1) and, potentially c2). However, considering that the cost of running the chain is either $2\gamma$ (when $\alpha\geq 2$) or $\alpha \gamma$ (when $\alpha<2$), the case $\gamma>1/6$ is always going to be sub-optimal if $\alpha \geq 2$; which is why in Theorem \ref{mainthm} we only consider $\gamma > 1/6$ if $\alpha<2$.

We therefore want to make assumptions on $S^N$, $\gamma$ and $\alpha$  to guarantee that our analysis falls only into one of the good regimes, \textit{i.e.}, to guarantee that  $Q^N_{\alpha}$  is of order one (or tends to zero). Assumption \ref{AssCLT} and Assumption \ref{Ass2} below are enforced in order to guarantee that this is indeed the case.

To explain the nature of Assumption \ref{AssCLT} below, we recall the approximation \eqref{qaapprra} and  the expression for $R^N_{\alpha}$, namely:
\begin{align}
R^N_{\alpha} = & - 2 \frac{\ell^{\alpha-1}}{N^{(\alpha-1)\gamma}} \langle{\cCn^{1/2}z^N},{S^Nx^N}\rangle - 2 \frac{\ell^{2(\alpha-1)}}{N^{2(\alpha-1)\gamma}}\norcn{\tS^N x^N}^2\nonumber\\
& - \frac{\ell^{2\alpha-1}}{N^{(2\alpha-1)\gamma}}\lancn{ \tS^N \cCn^{1/2}z^N}{\tS^N x^N}- \frac{1}{2}\frac{\ell^{2(2\alpha-1)}}{N^{2(2\alpha-1)\gamma}} \norcn{\tSn^2 x^N}^2\nonumber
\\
&- \frac{\ell^{3\alpha-1}}{N^{(3\alpha-1)\gamma}} \lancn{\tS^N \cCn^{1/2}z^N}{\tSn^2 x^N} \end{align}

 Roughly speaking,  Assumption \ref{AssCLT} requires that a CLT should hold for all the scalar products in the above expressions, Assumption \ref{Ass2} requires  a Law of Large Numbers (LLN) to hold for the remaining terms. We will first state such assumptions and then comment on them in Remark \ref{remonassumptions}.

We recall that $\tS^N=\cC^N S^N$ so, since $\cC^N$ is fixed by the problem and invertible, the following assumptions can be equivalently rephrased in terms of either $\tS^N$ or $S^N$.
\begin{assumption}\label{AssCLT}
The entries of the matrix $S^N$ do not depend on $N$.
There exist positive constants $0\leq d_1, d_2, d_3 <\infty$ such that the sequence of  matrices $\{S^N\}_N$ (and the related $\{\tS^N\}_N$) satisfies the following conditions, for some $\alpha \geq 1$:
 \begin{align*}
\textrm{i)} & \quad \frac{1}{N^{(\alpha-1)\gamma}} \langle{(\cC^N)^{1/2}z^N},{S^Nx^N}\rangle
\stackrel{\D}{\longrightarrow} {\cN} (0,d_1)\\
\textrm{ii)} & \quad \frac{1}{N^{(2\alpha-1)\gamma}}\lancn{ \tS^N (\cC^N)^{1/2}z^N}{\tS^N x^N}\stackrel{\D}{\longrightarrow} \cN (0,d_2)\\
\textrm{iii)} & \quad \frac{1}{N^{(3\alpha-1)\gamma}} \lancn{\tS^N (\cC^N)^{1/2}z^N}{(\tS^N)^2 x^N}\stackrel{\D}{\longrightarrow} \cN (0,d_3),
\end{align*}
where $\stackrel{\D}{\longrightarrow} $ denotes convergence in distribution.
\end{assumption}

\begin{assumption}\label{Ass2} The parameter $\alpha\geq 1$ and the sequence of  matrices $\{S^N\}_N$ are such that
\begin{align}
 \textrm{i)}\, & \lim_{N \ra \infty}\Epin \lv \frac{\norcn{\tS^N x^N}^2}{N^{2(\alpha-1)\gamma}} - c_1 \rv^2
= 0  \label{f1cond}\\
\textrm{ii)}\, &\lim_{N \ra \infty}\Epin \lv \frac{\norcn{(\tS^N)^2 x^N}^2  }{N^{2(2\alpha-1)\gamma}} - c_2\rv^2 = 0 \nonumber\\
\textrm{iii)}\, &\lim_{N \ra \infty} \quad \EE_{\pi^N} \lv \frac{\norcn{(\tS^N)^T (\tS^N)^2x}}{N^{2(3\alpha-1)\gamma}} - c_3\rv^2 = 0\nonumber \, ,
\end{align}
for some constants $0\leq c_1, c_2,c_3 <\infty$. We do not consider sequences $\{S^N\}$, such that \eqref{f1cond} is satisfied with $c_1=0$ and $\alpha=1$. That is, we  consider sequences $\{S^N\}_N$ such that \eqref{f1cond} is satisfied with $c_1=0$,  only in the cases in which this happens for $\alpha>1$, see Remark \ref{remonassumptions}.
\end{assumption}

In view of \eqref{constc1=}-\eqref{constc2=}, Assumption \ref{Ass2} implies the following three facts
\begin{align}
&\lim_{N \rightarrow \infty} \frac{\Epin\norcn{\tS^N x^N}^2}{N^{2(\alpha-1)\gamma}} = c_1 \label{f1}\\
&\lim_{N \rightarrow \infty} \frac{\Epin\norcn{(\tS^N)^2 x^N}^2}{N^{2(2\alpha-1)\gamma}} = c_2 \label{f2}\\
& \lim_{N \rightarrow \infty}\EE_{\pi^N}\frac{\norcn{(\tS^N)^T (\tS^N)^2x^N}}{N^{2(3\alpha-1)\gamma}}=c_3 \,. \label{f3}
\end{align}

Let us state a sufficient condition under which Assumption \ref{AssCLT} holds.
\begin{con}\label{Ass1}{ Let $z^N \sim \cN(0, \mathrm{Id}_N)$ and $x^N\sim \pi^N$.
The sequence of  matrices $\{S^N\}_N$ is such that for every $h=1 \dd N$ the $\sigma$-algebra generated by the random variable
$ z^{h,N}(S^Nx^N)^h$ is independent of the $\sigma$-algebra generated by the sum
$$
\sum_{j=1}^{h-1}  z^{j,N}(S^Nx^N)^j.
$$
In the above we stress that $(S^Nx^N)$ is a vector in $\R^N$ and $(S^N x^N)^h$ is the $h$-th component of such a vector. In order for Assumption \ref{AssCLT} to be satisfied, the same should hold  for $(\tS^N)^T \tS^N$ and for
$(\tS^N)^T (\tS^N)^2$ as well.}
\end{con}
Notice that, because $\cC^N$ is diagonal, if $(S^N)^2$ and $(S^N)^3$ satisfy Condition \ref{Ass1},   then also $(\tS^N)^T \tS^N$ and   $(\tS^N)^T (\tS^N)^2$ do. We give an example of a matrix satisfying all such conditions in  Remark \ref{remonassumptions}. Finally, for technical reasons (see proof of statement (ii) of Lemma \ref{lemmapreliminaries}, we also assume the following.

\begin{assumption}\label{extrassT2}
\begin{enumerate}
\item  Under $\pi^N$, the variables appearing in \eqref{f1}- \eqref{f3} have exponential moments, i.e., there exists a constant $m$, independent of $N$, such that
$$
\Epin e^{\frac{\norcn{\tS^N x^N}^2}{N^{2(\alpha-1)\gamma}}},
\Epin e^{\frac{\norcn{(\tS^N)^2 x^N}^2}{N^{2(2\alpha-1)\gamma}}},
\Epin e^{\frac{\norcn{(\tS^N)^T (\tS^N)^2x^N}}{N^{2(3\alpha-1)\gamma}}} <  m.
$$
\item  There exists some $r>1$ such that
$$
\Epin \left( \sum_{j=1}^N \frac{\lambda_j^6 \lv (S^Nx^N)^j \rv^4}{N^{2\gamma(\alpha-1)}}\right)^r < \infty.
$$

\end{enumerate}

\end{assumption}

\begin{remark} \label{remonassumptions} \textup{ Let us clarify the meaning of Assumption \ref{AssCLT},  Assumption \ref{Ass2} and Condition \ref{Ass1}. }
\begin{itemize}
\item \textup{In Assumption \ref{AssCLT}, we do not allow the entries of $S^N$ to depend explicitly on $N$ because this can be achieved by simply changing scaling (\textit{i.e.}, values of $\alpha$). }
\item\textup{ Recall that in stationarity $x^{i,N} \sim \mathcal{N}(0, \lambda_i^2)$ and observe  that without any assumption on $S^N$ the addends in the sum
\be\label{clttbs}
\langle{(\cC^N)^{1/2}z^N},{S^Nx^N}\rangle= \sum_{i=1}^N { \lambda_i z^{i,N}(S^Nx^N)^i}= \sum_{i=1}^N {\lambda_i z^{i,N}}\sum_{j=1}^N S^N_{i,j}x^{j,N}
\ee
 are non-independent and non-identically distributed random variables,  so the Central Limit theorem does not necessarily hold and in general the behaviour of the above sum will depend on $S^N$.  However, if we want
$Q^N_{\alpha}$ to be order one, the validity of the CLT is a reasonable requirement to impose on $S^N$.  }
\item  \textup{Condition \ref{Ass1} allows  simplifications in the calculation of the asymptotic variance  for terms of the form \eqref{clttbs}. Indeed a generalized form of the CLT  (see \cite[Theorem 2.1]{Dvoretzky}) implies that (under suitable assumptions) the asymptotic variance of the sum   \eqref{clttbs} is related to the following sum of conditional expectations:
$$
\sum_{k}a_{k,N}^2, \quad {\mbox{where }}   a_{k,N}^2:= \EE\left(\lambda_k^2  {\left( z^{k,N}(S^Nx^N)^k\right)^2}\Big\vert \sum_{i=1}^{k-1} {z^{i,N}(S^Nx^N)^i}\right).
$$
This may not be easy to calculate in general.
If Condition \ref{Ass1} holds then clearly the above simplifies and  one can apply the Central Limit Theorem in its simplest form to  the  summation  \eqref{clttbs}.
  In particular, if Condition \ref{Ass1} holds then each term is independent of the sum of the previous terms so  we obtain that for large $N$ (and, again, assuming that we are in stationarity)
$$
\frac{1}{N^{(\alpha-1)\gamma}}\lancn{(\cC^N)^{1/2}z^N}{S^N x^N} \sim \mathcal{N}(0,{c_1}).
$$
Notice that we are requiring that also the structure of  $(S^N)^2$ and $(S^N)^3$ allows for a similar simplification. So the main effect of Condition \ref{Ass1} is to obtain $c_i=d_i$ for $i \in \{1,2,3\}.$ A matrix that satisfies all the requirements of Condition \ref{Ass1} is the matrix $S^N=[S^N_{ij}]_{i,j=1}^{N}$ of (\ref{graphjordan}). If, for example,  $S^N$ is defined by (\ref{graphjordan}) with $J_{i}=1$  then $(S^N)^2=-\mathrm{Id}$ so $(S^N)^3=-S^N$ and $(S^N)^4=\mathrm{Id}$.  }
\item \textup{Roughly speaking, while Assumption \ref{AssCLT} imposes the validity of a CLT, Assumption \ref{Ass2} requires  that a Law of Large numbers should hold.
This assumption  is again made to  guarantee that $Q^N_{\alpha}$ is of order one (or zero), which is the only case of interest, as we have argued (Because $Q^N_{\alpha} \simeq R^N_{\alpha}$, this is the same as assuming that  $R^N_{\alpha}$ is $O(1)$. )}
\item \textup{We only consider values of $\alpha$ in the range $\alpha\geq 1$ for consistency: the sequence $\Epin\norcn{\tS^N x^N}^2$  is positive (and non identically zero), and the entries of $S^N$ do not depend on $N$.  Therefore there is no $0\leq \alpha<1$ such that Assumption \ref{Ass2} is satisfied for some finite constant $c_1$.
}
\item \textup {Finally, we make the obvious observation that, under Assumption \ref{AssCLT}, Assumption \ref{Ass2} and Condition \ref{Ass1} if the constant $c_1$ is assumed to be vanishing, $c_1=0$, then also $c_2=c_3=0$. Therefore in this case $Q^N_{\alpha}$ tends to zero in $L^2$. \footnote{This does not mean that the effect of the irreversible term is destroyed. It simply means that the acceptance probability will not feel it.} If $c_1=0$ then we rule out the case $\alpha=1$ because this would imply $\Epin\norcn{\tS^N x^N} \ra 0$, which could be achieved with a different rescaling.  If at least one of these constants is non-zero then $Q^N_{\alpha}$ is $O(1)$.
}
\end{itemize}
\hf
\end{remark}
  To recap, we have that for $N$ large, then
  \be\label{QNapprgaussian}
  \bar{Q}^N \approx \mathcal{N}\left(-\frac{\ell^6}{32} \frac{\ell^6}{16}\right)
\mbox{ if }\gamma=1/6,  \qquad \bar{Q}^N \approx 0 \,\,
\mbox{ if }\gamma > 1/6 \,.
\ee
 On the other hand, if Assumption \ref{AssCLT} and  Assumption \ref{Ass2} hold, then, from \eqref{feqan}-\eqref{constc2=}, we deduce after some calculus \footnote{These calculations are a bit long but straightforward and follow the lines of the calculation done in Lemma \ref{variousestimates}, proof of point {\em vi)}.} that  asymptotically
\be\label{F}
Q^N_{\alpha} \approx \mathcal{N} (-\ta, \tb), \quad \gamma \geq 1/6,
\ee
where
\begin{align}
 \ta &:=  2 \ell^{2(\alpha-1)}c_1+ \frac{1}{2}\ell^{2(2\alpha-1)}c_2  \label{G}\\
\tb &:= 4 \ell^{2(\alpha-1)}d_1 +5
\ell^{2(2\alpha-1)}d_2 + \ell^{2(3\alpha-1)}d_3.\label{H} 
\end{align}
From this we can heuristically deduce that if $\tilde{Z}$ is a normally distributed random variable with
\begin{equation*}
\tilde{Z} \sim\mathcal{N} \left(-\frac{\ell^6}{32}-\ta, \frac{\ell^6}{16}+ \tb\right),
\end{equation*}
then, for $N$ large one has
\be\label{apprQN}
\bar{Q}^N+ Q^N_{\alpha}=Q^N \approx \tilde{Z}  \quad \mbox{if }\gamma=1/6, \qquad Q^N \approx \mathcal{N}(-\tilde{a}, \tilde{b}) \,\, \mbox{ if } \gamma>1/6 \,.
\ee
As explained in Remark \ref{remonassumptions} (see third bullet point), if also Condition \ref{Ass1} is enforced,  then $c_i=d_i$ for $i \in \{1,2,3\}$;  therefore, for $N$ large, one finds that $Q^N$ is approximately distributed like $Q$,
\be\label{asympnorQ}
Q^N \approx Q,
\ee
where
\be\label{asympQnormal1}
Q \sim\mathcal{N} \left(-\frac{\ell^6}{32}-a, \frac{\ell^6}{16}+ b\right) \,\,\, \mbox{ if } \gamma=1/6, \quad  Q \sim\mathcal{N}(-a,b) \,\,\, \mbox{ if } \gamma>1/6,
\ee
with
\begin{align}
 a &:=  2 \ell^{2(\alpha-1)}c_1+ \frac{1}{2}\ell^{2(2\alpha-1)}c_2 \label{alb}\\
b &:= 4 \ell^{2(\alpha-1)}c_1 +5
\ell^{2(2\alpha-1)}c_2 + \ell^{2(3\alpha-1)}c_3. \label{albb}
\end{align}

Before moving to the heuristic analysis of the drift coefficient, we set
\be\label{hS}
h\sub:= \EE \left[ (1 \wedge e^{Q}) \right] \, ,
\ee
with $Q$ as in \eqref{asympQnormal1}.
The constant $h\sub$ represents the limiting average acceptance probability and it can be calculated by using the following lemma.
\begin{lemma}\label{lemgaussmin1}
If $G$ is a normally distributed random variable with $G \sim \mathcal{N}(\mu, \delta^2)$ ($\mu, \delta \neq 0$) then
$$
\EE(1 \wedge e^{G})= e^{\mu + \delta^2/2} \Phi\left( -\frac{\mu}{\delta}- \delta \right)+ \Phi\left( \frac{\mu}{\delta} \right),
$$
and
$$
\EE(G\mathbf{1}_{\{G<0\}})= e^{\mu + \delta^2/2} \Phi\left( -\frac{\mu}{\delta}- \delta \right),
$$
where $\Phi$ is the CDF of a standard Gaussian.
\end{lemma}
\begin{proof} By direct calculation.
\end{proof}

Therefore,
\begin{align} \label{Eq:OptimalAcceptProb}
h\sub&= \EE(1 \wedge e^{{Q}}) {=}e^{\frac{b}{2}-a}\Phi\left(\frac{-\ell^{6}/32-b+a}{\sqrt{(\ell^{6}/16)+b}}\right)+\Phi\left(\frac{-\ell^{6}/32-a}{\sqrt{\ell^{6}/16+b}}\right), \quad \mbox{ if } \gamma=1/6, \,\, a,b,\in \R
\end{align}
and
$$
h\sub = e^{\frac{b}{2}-a} \Phi\left(\frac{a}{\sqrt{b}}-\sqrt{b}\right)+\Phi\left(- \frac{a}{\sqrt{b}}\right), \quad \mbox{ if } \gamma>1/6, \,\, a,b \neq 0 \,.
$$
The above discussion makes it clear that the acceptance probability $h\sub$ depends on the constants $c_{i}$ and therefore, ultimately,  on the choice of the antisymmetric matrix $S^N$. {Notice that if $S^{N}$ takes the Jordan block diagonal form then $c_{2}=c_{3}=0$ and $b=2a$ resulting in limiting acceptance probability $h\sub=h\sub^J$, as in \eqref{Eq:OptimalAcceptProbThm}.

 \subsection{Heuristic analysis of the drift coefficient}\label{sechad}
In view of the discussion in Section \ref{SS:HeurAcceptProb}, we comment separately on the case $\alpha>2$, $\alpha=2$ and $\alpha<2$.

$\bullet$ Let us first assume that $\gamma=1/6$ and suppose $\alpha > 2$. Then the approximate drift of the chain is given by \eqref{apprdrift1}-\eqref{apprdrift2}, where we take  $\zeta=2$. We will prove that when $\alpha>2$ the addend \eqref{apprdrift2} is asymptotically small. This happens because   $Q^N$ and $z^N$ are asymptotically independent; in particular we will show that   the correlations between $Q^N$ and $z^N$ decay faster than $N^{-1/6}$, that is
\begin{equation*}
N^{1/6}\EE_k \left[ \left(1 \wedge e^{Q^N}\right) \cC^{1/2} z_{k+1}^N \right] \ra 0
\end{equation*}
(the formal statement of the above heuristic is contained in \ref{epsx} and \ref{est11}. 
Therefore
\begin{align}
d^N(x_k^N) &\simeq N^{1/3}\EE_k \left[ (1 \wedge e^{Q^N}) \left(- \frac{\ell^2}{2N^{1/3}}  x_k^N + \frac{\ell^{\alpha}}{N^{\alpha/ 6}}\tilde{S}^{N}x_k^N \right) \right]
\nonumber \\
& = \EE_k \left[ (1 \wedge e^{Q^N}) \left(- \frac{\ell^2}{2}  x_k^N \right) \right]+ \EE_k
\left[ (1 \wedge e^{Q^N}) \frac{\ell^{\alpha}}{N^{(\alpha-2)/ 6}} \tilde{S}^{N}x_k^N  \right]. \label{apprdrift22}
\end{align}
 We now use  the definition \eqref{hS} and  the heuristic approximation \eqref{asympnorQ} and observe that if $\alpha>2$ the second addend in \eqref{apprdrift22} disappears in the limit;  therefore the drift coefficient of the limiting diffusion is
\begin{equation*}
d\sub(x) = - \frac{\ell^2}{2} h\sub x \,.
\end{equation*}

$\bullet$ If instead $\gamma=1/6$ and $\alpha=2$ then the second addend in \eqref{apprdrift22} does not disappear in the limit; on the contrary, with the same reasoning as in the previous case, it will converge to $h\sub \ell^2 \tS x$.
 Moreover,  the term in  \eqref{apprdrift2} is no longer asymptotically small and contributes to the limiting drift. The reason why this happens can be seen by recalling  that $Q^N= \bar{Q}^N+ Q^N_{\alpha}$: clearly, the correlations between $\bar{Q}^N$ and $z^N$ are not affected by the choice of $\alpha$; however the value of $\alpha$  does affect the decay of the  correlations between $Q^N_{\alpha}$ and $z^N$. \par
In \ref{lemmapreliminaries}  we present a calculation which shows that (under appropriate conditions on $S^N$), if $\alpha=2$ then
\begin{equation*}
N^{1/6}\ell\EE_k \left[ \left(1 \wedge e^{Q^N} \right)\cC^{1/2} z_{k+1}^N \right] \simeq -2 \ell^2 \nu\sub\tilde{S} x_k
\end{equation*}
where  $\nu\sub$ is a real constant, namely
\be\label{bfh}
\nu\sub := \EE e^{Q}\mathbf{1}_{\{Q<0\}} \stackrel{\mbox{if }\gamma=1/6}{=}
e^{\frac{b}{2}-a}\Phi\left(\frac{-\ell^{6}/32-b+a}{\sqrt{(\ell^{6}/16)+b}}\right) . \footnote{Having used Lemma \ref{lemgaussmin1} for the last equality.}
\ee
Thus we find that the drift of the limiting diffusion is
\begin{equation*}
d(x) = - \frac{\ell^2}{2} h\sub \,x + \tau\sub\,\ell^{2} \tilde{S}x ,
\end{equation*}
where
\begin{equation}\label{mhs1}
\tau\sub:= h\sub - 2 \nu\sub .
\end{equation}
 By Lemma \ref{lemgaussmin1}  one can also see that, again if $\gamma=1/6$,
\be
\tau\sub=
  -h\sub + 2 \Phi \left( \frac{-\frac{\ell^6}{32}-a}{\sqrt{\frac{\ell^6}{16}+b}} \right).
\ee
 If  the constant $c_1=0$ or if more generally $b=2a$ (e.g., in the Jordan block diagonal form case), then $\tau\sub=0$ and the limiting drift reduces to
\begin{equation*}
d(x) = - \frac{\ell^2}{2} h\sub \,x.
\end{equation*}
Notice that the constant $ h\sub$ can never vanish in this regime.

$\bullet$ As already observed, if $\gamma\geq 1/6$ and $1\leq \alpha<2$ we need to scale the algorithm differently; in particular, the approximate drift is given this time by
\begin{align}
d\sub^N(x_k^N) &\simeq N^{\alpha\gamma}\EE_k \left[ (1 \wedge e^{Q^N}) \left(- \frac{\ell^2}{2N^{2\gamma}}  x_k^N + \frac{\ell^{\alpha}}{N^{\alpha\gamma}} \tilde{S}^{N}x_k^N \right) \right]
+ N^{\alpha\gamma}\EE_k \left[ (1 \wedge e^{Q^N}) \frac{\ell}{N^{\gamma}}(\cC^N)^{1/2}z_{k+1}^N \right]
\label{apprdrift23}
\end{align}
Therefore in this case ($\alpha<2$) the first term in \eqref{apprdrift23} asymptotically disappears. As for the third addend, this is again asymptotically not small (see Lemma \ref{lemmapreliminaries}, part (ii) of \ref{AppendixB}) and the second addend contributes with a term $\ell^{\alpha} h\sub \tS x$.  Therefore, the  limiting drift is
$$
d(x)= \tau\sub \ell^{\alpha}\tilde{S}x,
$$
with $\tau\sub$ defined as in \eqref{mhs1}. Again, if  $b=2a$ (which happens e.g. for $c_1=c_2=c_3=0$), then $\tau\sub=0$.  Therefore, in this case the  limit of the chain would simply be $\dot{x}=0$.

\subsection{Heuristic analysis of the diffusion coefficient}\label{SS:HeurDiffCoeff}
The heuristic argument for approximate diffusion coefficient is relatively straightforward.
\begin{align}
M_k^N&= N^{\zeta\gamma/2}\left[x_{k+1}^N - x_k^N - \EE_k(x_{k+1}^N - x_k^N) \right]\nonumber\\
&=N^{\zeta\gamma/2}\left[\tilde{\beta}^{N}(y_{k+1}^N - x_k^N) - \EE_k[\tilde{\beta}^{N}(y_{k+1}^N - x_k^N)] \right]\nonumber\\
&=N^{\zeta\gamma/2}\left[\tilde{\beta}^{N}\left(- \frac{\ell^{2}}{2N^{2\gamma}} x_k^N+ \frac{\ell^{\alpha}}{2N^{\alpha\gamma}} \tilde{S}^{N} x_k^N+ \frac{\ell}{N^{\gamma}} \cC^{1/2}z_{k+1}^N\right) \right.\nonumber\\
&\qquad\qquad\left.- \EE_k[\tilde{\beta}^{N}\left(- \frac{\ell^{2}}{2N^{2\gamma}} x_k^N+ \frac{\ell^{\alpha}}{2N^{\alpha\gamma}} \tilde{S}^{N} x_k^N+ \frac{\ell}{N^{\gamma}} \cC^{1/2}z_{k+1}^N\right)] \right]\nonumber\\
&=N^{(\zeta/2-2)\gamma}\left(-\frac{\ell^{2}}{2}\tilde{\beta}^{N}x_k^N\right) + N^{(\zeta/2-\alpha)\gamma}\left(\ell^{\alpha}\tilde{\beta}^{N} \tilde{S}^{N} x_k^N\right)+
N^{(\zeta/2-1)\gamma}\left(\ell\tilde{\beta}^{N} \cC^{1/2}z_{k+1}^N\right)\nonumber\\
&\qquad\qquad- N^{-\zeta\gamma/2}d^{N}(x_k^N)\label{Eq:GeneralMartingaleDiff}
\end{align}

By Lemma \ref{lemma2}, we get that in the case  $\gamma=1/6$, $\alpha\geq2$ and $\zeta=2$,  $M_k^N$  behaves like $N(0,\ell^{2}h\sub \cC)$;  in the case $\gamma\geq1/6$, $\alpha<2$ and $\zeta=\alpha$ and $M_k^N$  converges to  zero.

\section{On numerical computation of acceptance probability} \label{S:NumericalApproxAccProb}
In this section we aim to show via concrete examples how to construct sequences of matrices $S^{N}$ that satisfy the assumptions of the paper.  In addition we calculate in these examples the constants $d_{i},c_{i}$ for $i=1,2,3$ that are necessary in order to numerically approximate the acceptance probabilities in each case.

Let us first discuss how one can construct $S^{N}$ so that it satisfies the assumptions of the paper, but at the same time leads to non-zero constant $c_{1}$.  Let us again consider a Jordan block-diagonal form for $S^{N}$ of (\ref{graphjordan}) but now set  $S_{1}^{N}=[S_{ij}]_{i,j=1}^{N}$ with
$S_{ij}=J_{i}$ and $S_{ji}=-J_{i}$ for $j=i+1$ and $i=1,3,5,7,\cdots,N-1$ and $S_{ij}=0$ otherwise with $J_{i}$ to be chosen.

Let us set $\lambda_{j}=j^{-k}$ for $k>1/2$. Then, we can compute
\begin{align}
\Epin  \norc{\tS_{1} x^N}^2&=\sum_{i=1}^{N}\lambda_{i}^{2}\sum_{j=1}^{N}(S_{ij})^{2}\lambda_{j}^{2}=\sum_{i \textrm{ odd }} i^{-2k}|S_{i,i+1}|^2(i+1)^{-2k}+\sum_{i \textrm{ even }} i^{-2k}|S_{i,i-1}|^2(i-1)^{-2k}\nonumber\\
&=2\sum_{i=1}^{N}(2i-1)^{-2k}J_{i}^{2}(2i)^{-2k}\nonumber
\end{align}

In order to make sure that $c_1\neq0$, let us choose $J_{i}=(2i-1)^k (2i)^k i^{[2(\alpha-1)\gamma-1]/2}$ with $\alpha>1$. Then, we have that $\Epin  \norc{\tS_{1} x^N}^2=2\sum_{i=1}^{N}i^{2(\alpha-1)\gamma-1}$. We obtain that
\begin{align}
c_{1}&=\lim_{N \ra \infty}\Epin  \frac{\norcn{\tS_{1} x^N}^2}{N^{2(\alpha-1)\gamma}}=\lim_{N \ra \infty}  \frac{1}{N} \frac{2\sum_{i=1}^{N}\left(\frac{i}{N}\right)^{2(\alpha-1)\gamma-1}}{N^{2(\alpha-1)\gamma}}=2\int_{0}^{1}x^{2(\alpha-1)\gamma-1}dx=\frac{1}{(\alpha-1)\gamma}.\nonumber
\end{align}

Next we compute $c_{2}$ for this particular choice of the $S$ matrix. We have
\begin{align}
\Epin  \norcn{\tS_{1}^2 x^N}^2  &=\sum_{i=1}^{N}\lambda_{i}^{2}\sum_{j=1}^{N}((S\tilde{S})_{ij})^{2}\lambda_{j}^{2}=\sum_{i=1}^{N}\lambda_{i}^{2}((S\tilde{S})_{ii})^{2}\lambda_{i}^{2}=2\sum_{i \textrm{ odd }}^{N}\lambda_{i}^{4}\lambda_{i+1}^{4}J_{(i+1)/2}^{4}\nonumber\\
&=2\sum_{i \textrm{ odd }}^{N}i^{-4k}(i+1)^{-4k} (2(i+1)/2-1)^{4k} (2(i+1)/2)^{4k} \left(\frac{i+1}{2}\right)^{2[2(\alpha-1)\gamma-1]}\nonumber\\
&=2\sum_{i \textrm{ odd }}^{N}i^{-4k}(i+1)^{-4k} (i)^{4k} (i+1)^{4k} \left(\frac{i+1}{2}\right)^{2[2(\alpha-1)\gamma-1]}\nonumber\\
&=2\sum_{i \textrm{ odd }}^{N} \left(\frac{i+1}{2}\right)^{2[2(\alpha-1)\gamma-1]}\nonumber\\
&=2\sum_{i=1}^{N} i^{2[2(\alpha-1)\gamma-1]}\nonumber
\end{align}

Similarly to the computation for $c_{1}$, we obtain that $c_{2}=0$. As far as $c_{3}$ is concerned, we obtain by a computation similar to the computation for $c_{2}$ that
\begin{align}
c_{3}&=\lim_{N \ra \infty}\EE_{\pi^N}  \frac{\norc{\tS_{1}^T \tS_{3}^2x}}{N^{2(3\alpha-1)\gamma}}=0\nonumber
\end{align}

For the same reasons as in Remark  \ref{remonassumptions} we also have $d_{i}=c_{i}$ for $i=1,2,3$. Hence, we have obtained that in this example $d_{1}=c_{1}=\frac{1}{(\alpha-1)\gamma}$ and $d_{2}=c_{2}=d_{3}=c_{3}=0$.
Then, this implies that  $Q \sim\mathcal{N} \left(-\frac{\ell^6}{32}-a, \frac{\ell^6}{16}+b\right)$ in (\ref{asympQnormal1}), with $a=2\ell^{2(\alpha-1)}\frac{1}{(\alpha-1)\gamma}$ and $b=4\ell^{2(\alpha-1)}\frac{1}{(\alpha-1)\gamma}$. Hence, in this case $b=2a$ and thus $\tau\sub=0$.

Next we discuss optimal acceptance probabilities. Since in this case $b=2a$ which then implies $\tau\sub=0$, we can answer this question in the case of regimes i), i.e., $\alpha>2$ and ii) $\alpha=2$. In these cases we can find the $\ell$ that maximizes the speed of the limiting diffusion, i.e., $\ell^{2}h\sub$ and then for the maximizing $\ell$ compute the limiting acceptance probability $h\sub$ as given by (\ref{Eq:OptimalAcceptProb}).
In Table \ref{T:OptimalAccProb} we record optimal acceptance probabilities $h\sub$ derived in (\ref{Eq:OptimalAcceptProb}) for different choices for $\alpha$ when $\gamma=1/6$ in the case the matrix $S$ being $S_{1}$ as defined above.

\begin{table}[th]
  \begin{tabular}[c]{|c||c|c|c|c|c|c|c|}
    \hline
     $\alpha$ & $2$ & $4$ & $6$  & $8$  & $10$ & $15$ & $30$\\
    \hline
    $\textrm{optimal }  h\sub$  & $0.234$ & $0.574$ & $0.702$  & $0.767$ &$0.803$ & $0.848$ & $0.884$\\

       \hline
\end{tabular}
\medskip
\caption{Optimal acceptance probability when $S^{N}=S^{N}_{1}$ for different $\alpha'$s when $\gamma=1/6$.} \label{T:OptimalAccProb}
\end{table}

In particular, Table \ref{T:OptimalAccProb} shows that by increasing $\alpha$, we can increase the optimal acceptance probability as desired.

Next, we discuss two cases where $c_{1}=c_{2}=c_{3}=0$. In particular, one can define the matrix $S_{2}^{N}=[S_{ij}]_{i,j=1}^{N}$  such that  $S_{ij}=1$ and $S_{ji}=-1$ for $j=i+1$ and $i=1,2,3,4,\cdots,N$ and $S_{ij}=0$ otherwise. A similar alternative construction is to set  $S_{3}^{N}=[S_{ij}]_{i,j=1}^{N}$ with $S_{ij}=1$ and $S_{ji}=-1$ for $j=i+1$ and $i=1,3,5,\cdots,N-1$ and $S_{ij}=0$ otherwise. In both cases $S_{2}^{N}$ and $S_{3}^{N}$ one can check that $c_{i}=d_{i}=0$ for all $i=1,2,3$. This then implies that $Q \sim\mathcal{N} \left(-\frac{\ell^6}{32}, \frac{\ell^6}{16}\right)$ in (\ref{asympQnormal1}), as $a=b=0$ in this case. Hence, one recovers that in this case the optimal acceptance probability for both choices $S_{1}^{N}$ and $S_{2}^{N}$ is achieved at $0.574$ as in the classical work \cite{RobertsRosenthal1998}.

In Section \ref{Eq:SimulationStudies} we present simulation studies based on using $S_{1}^{N}$, $S_{2}^{N}$ and $S_{3}^{N}$ confirming and illustrating the theoretical results.

\section{Simulation studies}\label{Eq:SimulationStudies}
The goal of this section is to illustrate the theoretical findings of this paper via simulation studies. Upon close inspection of Theorem \ref{mainthm} it becomes apparent that closed form expressions of quantities such as optimal acceptance probabilities and optimal scalings are hard to obtain and dependent on the choice of the sequence $S^{N}$. For this reason we will resort to simulations.

As a measure of performance we choose similarly to \cite{RobertsRosenthal1998,RobertsRosenthal2001} the expected squared jumping distance (ESJD) of the algorithm defined as
\[
\mathrm{ESJD}=\EE\left|X^{1}_{k+1}-X^{1}_{k}\right|^{2}
\]

It is easy to see that maximizing ESJD is equivalent to minimizing the first order lag autocorrelation function, $\rho_{1}=\text{Corr}_{\pi}(X_{k+1},X_{k})$. Hence, motivated by \cite{RobertsRosenthal2001} and based on Theorem \ref{mainthm} we define the computational time as
\begin{align}
CT\sub=-\frac{1}{N^{\zeta\gamma}\log \rho_{1}}\label{Eq:CT}
\end{align}
where $\zeta=2$ if $\alpha\geq 2$ and $\zeta=\alpha$ if $\alpha<2$. 
 We seek to minimize $CT\sub$. A word of caution is necessary here. Even though $CT\sub$ as a measure of performance is easily justifiable in the case that one has a diffusion limit, \textit{i.e.}, when $\alpha\geq 2$, the situation is more complex for $\alpha<2$. By Theorem \ref{mainthm} if $\alpha<2$, the algorithm converges to a deterministic ODE limit and not to a diffusion limit. Hence, in this case the absolute jumping distance $\left|X^{1}_{k+1}-X^{1}_{k}\right|$ may make more sense as a measure of performance in the limit as $N\rightarrow\infty$. However, in the deterministic case, maximizing the absolute jumping distance or the squared jumping distance are equivalent tasks. Hence, working with the acceptance probability and with ESJD and as a consequence with $CT\sub$ as defined in (\ref{Eq:CT}) may be reasonable criterion on quantifying convergence.

The simulations presented below were performed using a MPI C parallel code. The numerical results that follow report sample averages and sample standard deviations of quantities such as acceptance probability and different observables. The reported values were obtained based on $1600$ independent repetitions of MALA and \imala\, algorithm with $10^4$ steps for each independent realization. For each independent realization of the algorithm the numerical values selected are the ones corresponding to minimizing $CT\sub$. We also remark here that in all of the simulation studies $CT\sub$ was a convex function of the acceptance probability (this resembles the behavior in the standard MALA case, see \cite{RobertsRosenthal2001}).

Below we report a statistical estimator and its corresponding empirical standard deviation for the optimal acceptance probability. In addition, to demonstrate convergence to the equilibrium we also computed a few observables. We report below statistical estimators along with their corresponding statistical standard deviations for
\[
\theta_{2}=\EE_{\pi^{N}}\|X\|^{2}, \text{ and } \theta_{3}=\EE_{\pi^{N}}\sum_{i=1}^{N}X_{i}^{3},
\]
where with some abuse of notation we have denoted by $X_{i}$ the $i^{th}$ component of the $X$ vector. In regards to the measure $\pi^{N}\sim N(0,\mathcal{C}^{N})$, where $\mathcal{C}^{N}=\text{diag}(\lambda_{1}^{2},\cdots,\lambda_{N}^{2})$ with $\lambda_{i}^{2}=i^{-2}$.

Let us first study the behavior of the algorithm when $S^{N}=S^{N}_{1}$, as defined in Section \ref{S:NumericalApproxAccProb}. In this case $c_{1}=\frac{1}{(\alpha-1)\gamma}\neq 0 $ and $c_{2}=c_{3}=0$. However, as we saw in Table \ref{T:OptimalAccProb} the effect of the value of $\alpha$ on the optimal acceptance probabilities is significant.

Let us demonstrate the behavior via a number of simulation studies reported in Tables \ref{Table_d10_3}-\ref{Table_d100_3}. As we described in Section \ref{S:NumericalApproxAccProb}, such a computation makes sense mainly in the case $\alpha\geq 2$ and $\gamma=1/6$. Of course one can compute via simulation empirical values for any value of $\alpha>1$ and $\gamma>0$. The acceptance probabilities in the case $\alpha\in(1,2)$ were very small, leading to estimators with large variances. Hence, we only report values for larger values of $\alpha$ where one may expect to have results that are comparable to what standard MALA gives. For completeness,  we also present data for $\gamma=1/2$, even though for standard MALA the optimal choice is $\gamma=1/6$.

Notice that the estimates for acceptance probabilities $\widehat{h}\sub$ align very well with the theoretical predictions that appear in Table \ref{T:OptimalAccProb}. In particular, as $\alpha$ increases the optimal acceptance probability also increases, in accordance to what the theory predicts.
\begin{table}[htbp!]
\begin{center}%
\begin{scriptsize}
\begin{tabular}
[c]{|c|c|c|c|c|}\hline
$\gamma=1/6$ & $(\widehat{h}\sub, \widehat{sd(h\sub)})$ &$(\widehat{\theta}_{2}, \widehat{sd(\theta_{2})})$&$(\widehat{\theta}_{3}, \widehat{sd(\theta_{3})})$ &$\widehat{CT}\sub$   \\
\hline  \text{standard MALA}& $(0.541, 0.043)$& $(1.073, 0.022)$& $(-0.001, 0.061)$& $0.725$\\
        $\alpha=2$ & $(0.229,0.056)$ &  $(0.935,0.081)$   & $(-0.001, 0.211)$   & $8.139$      \\
        $\alpha=4$ & $(0.566,0.065)$ &  $(1.026,0.034)$   & $(0.006, 0.141)$   & $3.367$      \\
        $\alpha=6$ & $(0.686,0.063)$ &  $(1.045,0.072)$   & $(0.003, 0.101)$   & $1.946$    \\
        $\alpha=8$ & $(0.751, 0.051)$&  $(1.051,0.026)$   & $(0.0001, 0.101)$   & $1.418$     \\
        $\alpha=10$ & $(0.784,0.052)$ &  $(1.056, 0.061)$  & $(0.004, 0.172)$   & $1.151$     \\
        $\alpha=15$ & $(0.821, 0.045)$ &  $(1.059, 0.031)$ & $(-0.002,0.096)$   & $0.951$      \\
        $\alpha=30$ & $(0.855, 0.026)$&  $(1.069, 0.025)$   & $(-0.003, 0.087)$   & $0.798$    \\
\hline
$\gamma=1/2$ &  & &  &   \\
\hline  \text{standard MALA}& $(0.753, 0.008)$& $(1.072, 0.020)$& $(-0.002, 0.068)$& $0.198$\\
        $\alpha=2$ & $(0.205,0.066)$ &  $(0.869,0.118)$   & $(0.0002, 0.339)$   & $9.419$      \\
        $\alpha=4$ & $(0.555,0.092)$ &  $(0.979,0.087)$   & $(-0.012, 0.242)$   & $2.022$      \\
        $\alpha=6$ & $(0.691,0.083)$ &  $(0.997,0.056)$   & $(0.002, 0.174)$   & $1.411$    \\
        $\alpha=8$ & $(0.761, 0.068)$&  $(1.001,0.038)$   & $(-0.008, 0.166)$   & $1.232$     \\
        $\alpha=10$ & $(0.808,0.064)$ &  $(1.015, 0.036)$  & $(0.002, 0.168)$   & $1.051$     \\
        $\alpha=15$ & $(0.866, 0.054)$ &  $(1.023, 0.032)$ & $(-0.0003,0.156)$   & $0.551$      \\
        $\alpha=30$ & $(0.924, 0.046)$&  $(1.038, 0.067)$   & $(-0.002, 0.206)$   & $0.213$    \\
\hline
\end{tabular}
\end{scriptsize}
\end{center}
\caption{Dimension $N=10$ and $S^{N}=S^{N}_{1}$. }
\label{Table_d10_3}%
\end{table}

\begin{table}[htbp!]
\begin{center}%
\begin{scriptsize}
\begin{tabular}
[c]{|c|c|c|c|c|}\hline
$\gamma=1/6$ & $(\widehat{h}\sub, \widehat{sd(h\sub)})$ &$(\widehat{\theta}_{2}, \widehat{sd(\theta_{2})})$&$(\widehat{\theta}_{3}, \widehat{sd(\theta_{3})})$ &$\widehat{CT}\sub$   \\
\hline  \text{standard MALA}& $(0.538, 0.052)$& $(1.054, 0.025)$& $(0.004, 0.089)$& $0.912$\\
        $\alpha=2$ & $(0.216,0.077)$ &  $(0.882,0.177)$   & $(-0.035, 0.451)$   & $8.453$      \\
        $\alpha=4$ & $(0.563,0.074)$ &  $(1.004,0.041)$   & $(0.004, 0.180)$   & $3.143$      \\
        $\alpha=6$ & $(0.689,0.068)$ &  $(1.026,0.035)$   & $(-0.009, 0.141)$   & $2.024$    \\
        $\alpha=8$ & $(0.748, 0.064)$&  $(1.036,0.042)$   & $(0.003, 0.157)$   & $1.512$     \\
        $\alpha=10$ & $(0.784,0.063)$ &  $(1.048, 0.098)$  & $(0.006, 0.181)$   & $1.269$     \\
        $\alpha=15$ & $(0.824, 0.045)$ &  $(1.048, 0.031)$ & $(-0.005,0.123)$   & $0.991$      \\
        $\alpha=30$ & $(0.849, 0.035)$&  $(1.054, 0.025)$   & $(-0.0002, 0.114)$   & $0.843$    \\
\hline
$\gamma=1/2$ &  & &  &   \\
\hline  \text{standard MALA}& $(0.849, 0.002)$& $(1.038, 0.037)$& $(0.0004, 0.147)$& $0.211$\\
        $\alpha=2$ & $(0.222,0.102)$ &  $(0.681,0.111)$   & $(0.025, 0.437)$   & $3.396$      \\
        $\alpha=4$ & $(0.564,0.013)$ &  $(0.837,0.107)$   & $(0.008, 0.378)$   & $1.696$      \\
        $\alpha=6$ & $(0.699,0.023)$ &  $(0.887,0.139)$   & $(0.003, 0.272)$   & $1.405$    \\
        $\alpha=8$ & $(0.774, 0.071)$&  $(1.001,0.019)$   & $(0.005, 0.267)$   & $1.269$     \\
        $\alpha=10$ & $(0.816,0.034)$ & $(0.967, 0.056)$  & $(-0.002, 0.263)$   & $1.101$     \\
        $\alpha=15$ & $(0.873, 0.088)$ &  $(0.948, 0.101)$ & $(0.008,0.264)$   & $0.781$      \\
        $\alpha=30$ & $(0.932, 0.073)$&  $(0.958, 0.071)$   & $(-0.001, 0.114)$   & $0.583$    \\
\hline
\end{tabular}
\end{scriptsize}
\end{center}
\caption{Dimension $N=50$ and $S^{N}=S^{N}_{1}$. }
\label{Table_d50_3}%
\end{table}

\begin{table}[htbp!]
\begin{center}%
\begin{scriptsize}
\begin{tabular}
[c]{|c|c|c|c|c|}\hline
$\gamma=1/6$ & $(\widehat{h}\sub, \widehat{sd(h\sub)})$ &$(\widehat{\theta}_{2}, \widehat{sd(\theta_{2})})$&$(\widehat{\theta}_{3}, \widehat{sd(\theta_{3})})$ &$\widehat{CT}\sub$   \\
\hline  \text{standard MALA}& $(0.538, 0.057)$& $(1.047, 0.036)$& $(0.001, 0.093)$& $0.938$\\
        $\alpha=2$ & $(0.225,0.071)$ &  $(0.891,0.178)$   & $(0.003, 0.352)$   & $7.654$      \\
        $\alpha=4$ & $(0.559,0.083)$ &  $(1.003,0.039)$   & $(-0.008, 0.180)$   & $3.467$      \\
        $\alpha=6$ & $(0.688,0.073)$ &  $(1.019,0.047)$   & $(-0.001, 0.168)$   & $2.012$    \\
        $\alpha=8$ & $(0.749, 0.068)$&  $(1.032,0.036)$   & $(-0.006, 0.155)$   & $1.488$     \\
        $\alpha=10$ & $(0.784,0.059)$ &  $(1.034, 0.038)$  & $(-0.001, 0.144)$   & $1.245$     \\
        $\alpha=15$ & $(0.823, 0.051)$ &  $(1.043, 0.038)$ & $(-0.008,0.131)$   & $0.981$      \\
        $\alpha=30$ & $(0.851, 0.034)$&  $(1.053, 0.033)$   & $(-0.0016, 0.124)$   & $0.942$    \\
\hline
$\gamma=1/2$ &  & &  &   \\
\hline  \text{standard MALA}& $(0.948, 0.002)$& $(1.019, 0.045)$& $(0.0001, 0.183)$& $0.211$\\
        $\alpha=2$ & $(0.221,0.101)$ &  $(0.782,0.113)$   & $(0.015, 0.357)$   & $3.145$      \\
        $\alpha=4$ & $(0.562,0.015)$ &  $(0.843,0.139)$   & $(0.012, 0.313)$   & $1.551$      \\
        $\alpha=6$ & $(0.696,0.034)$ &  $(0.841,0.121)$   & $(0.001, 0.275)$   & $1.299$    \\
        $\alpha=8$ & $(0.771, 0.082)$&  $(0.951,0.082)$   & $(-0.001, 0.217)$   & $1.187$     \\
        $\alpha=10$ & $(0.817,0.067)$ & $(0.989, 0.063)$  & $(-0.003, 0.311)$   & $1.127$     \\
        $\alpha=15$ & $(0.873, 0.093)$ &  $(1.003, 0.102)$ & $(-0.002,0.235)$   & $0.967$      \\
        $\alpha=30$ & $(0.937, 0.075)$&  $(0.908, 0.101)$   & $(-0.013, 0.111)$   & $0.486$    \\
\hline
\end{tabular}
\end{scriptsize}
\end{center}
\caption{Dimension $N=100$ and $S^{N}=S^{N}_{1}$. }
\label{Table_d100_3}%
\end{table}

To visualize the situation, in Figure \ref{figsim3} we have plotted $\alpha$ versus the empirical acceptance probabilities $\widehat{h}\sub$ for $\gamma=1/6$.  As expected by the theory,  as $\alpha$ increases, $\widehat{h}\sub$ increases.
\begin{figure}[!h]
 \centering
\begin{tabular}{c}
 \includegraphics[scale=0.3]{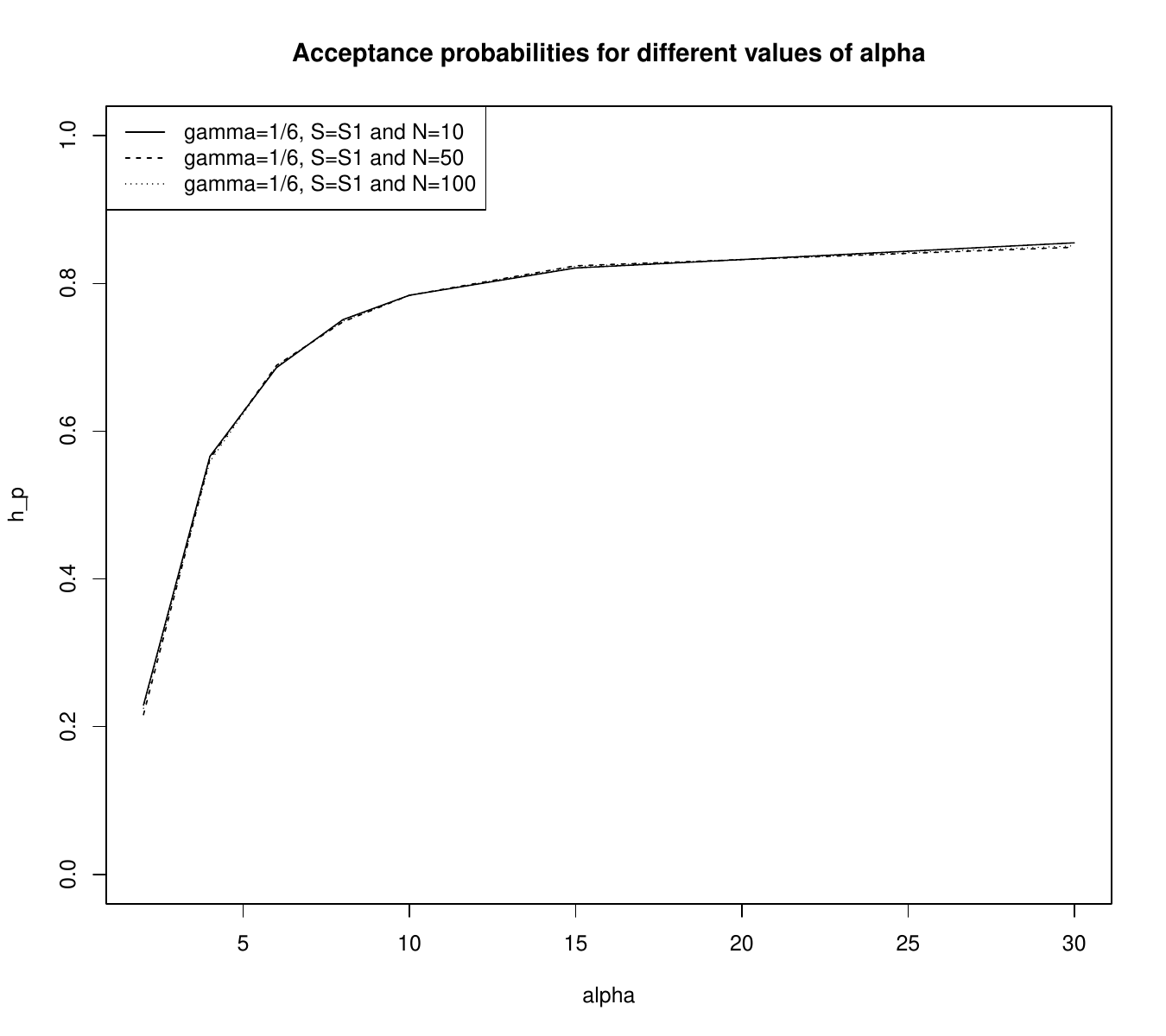}
  \end{tabular}
\caption{Empirical acceptance probability $\widehat{h}\sub$ versus values of $\alpha$ for $S^{N}=S^{N}_{1}$  and $\gamma=1/6$. For reference the limiting optimal acceptance probability for standard MALA is $0.574$.} \label{figsim3}
 \end{figure}

In Figure \ref{figsim2} we have plotted $\alpha$ versus the statistic $\widehat{CT}\sub$. In each case, standard MALA is represented by the rightmost value for $\alpha$. As expected by the theory, in each of the cases, as $\alpha$ increases $\widehat{CT}\sub$ converges to the corresponding value of standard MALA.
\begin{figure}[!h]
 \centering
\begin{tabular}{cc}
 \includegraphics[scale=0.3]{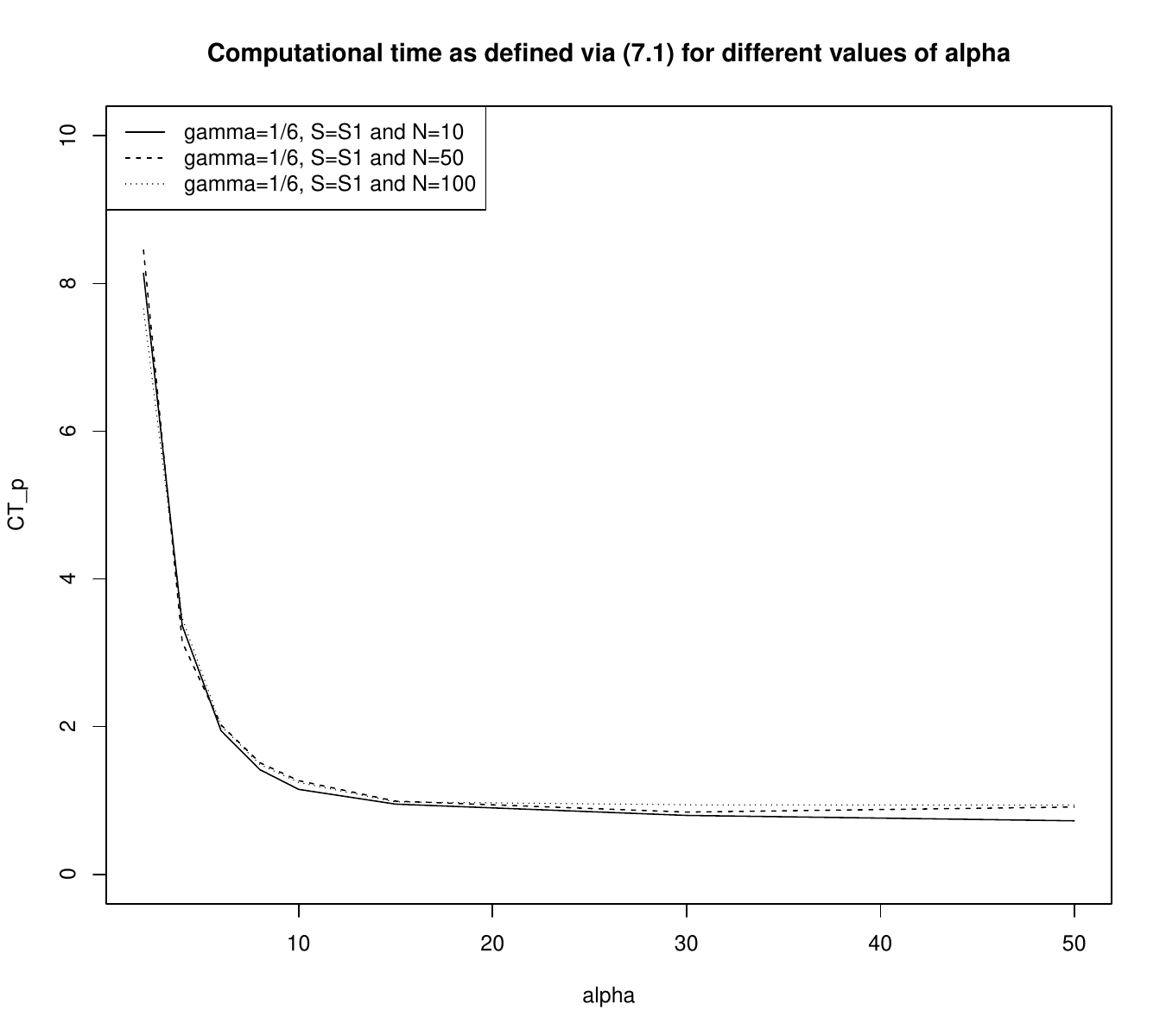}
  &
  \includegraphics[scale=0.3]{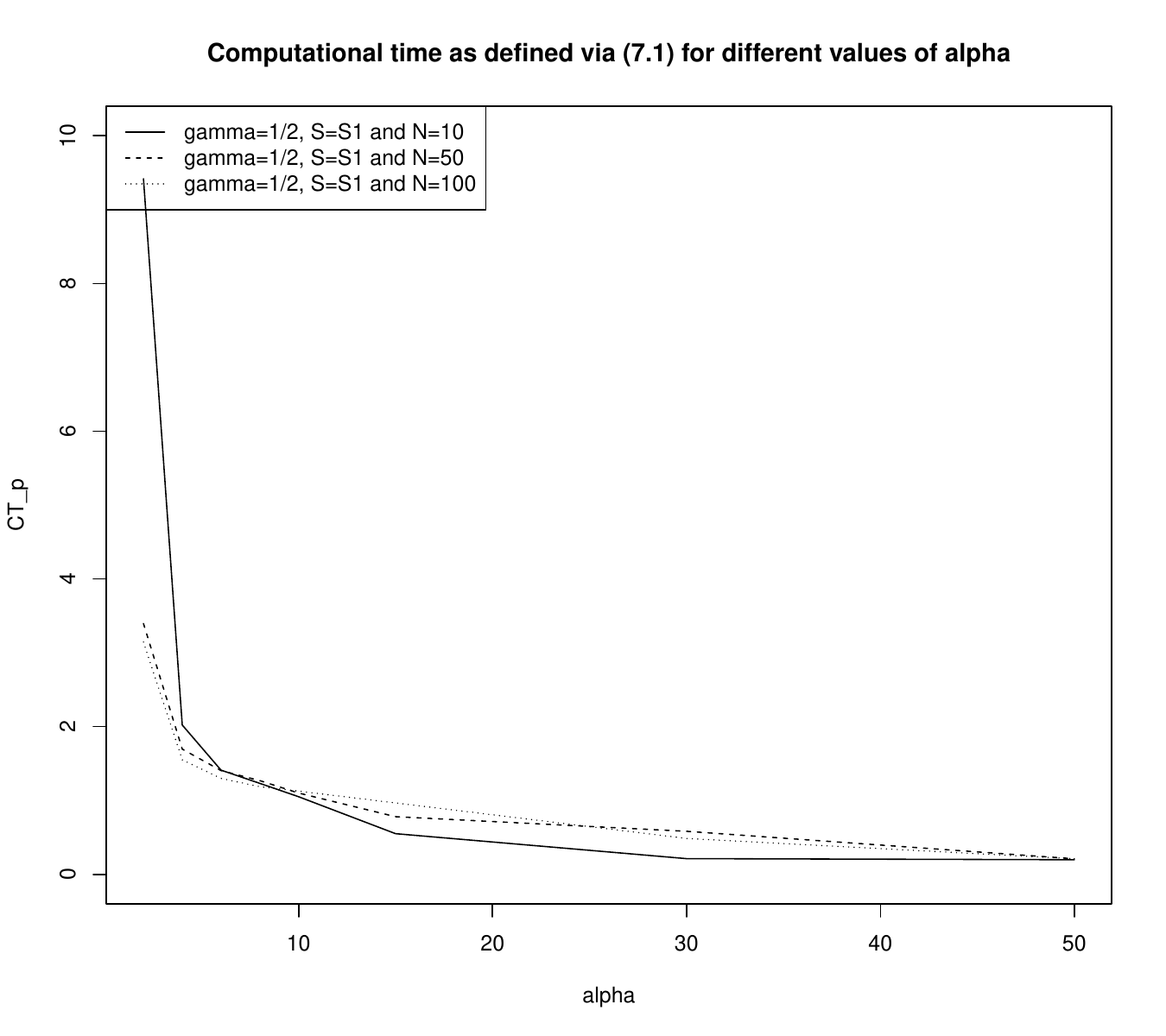}\\
 (a) $S^{N}=S^{N}_{1}$  and $\gamma=1/6$ &  (b) $S^{N}=S^{N}_{1}$  and $\gamma=1/2$\\
 \end{tabular}
\caption{Measure $\widehat{CT}\sub$ versus values of $\alpha$. In each of the figures the rightmost value for $\alpha$ corresponds to standard MALA} \label{figsim2}
 \end{figure}

 Next we record simulations for the antisymmetric matrix $S^{N}=S^{N}_{2}$, defined in Section \ref{S:NumericalApproxAccProb}.  We remark here that we also did the same simulations with the $S=S^{N}_{3}$ and the numerical results were statistically the same. So, we only report the ones for $S^{N}=S^{N}_{2}$, see Tables \ref{Table_d10}-\ref{Table_d100}.

In the case of $S^{N}=S^{N}_{2}$, we have that $c_{i}=d_{i}=0$ and thus the optimal acceptance probability in the limit as $N\rightarrow\infty$ is the same as with standard MALA, i.e. around $0.574$. However, as the Tables \ref{Table_d10}-\ref{Table_d100} demonstrate when $\gamma=1/6$ (the optimal choice for standard MALA) this is being realized for much larger values of $N$ for \imala\ compared to MALA. On the other hand if we take $\gamma>1/6$ (which is not optimal for standard MALA), e.g., $\gamma=1/2$ in Table \ref{Table_d100}, and focus on the fluid regime, regime iii), we notice that $CT\sub$ is estimated to be smaller with \imala\ compared to MALA, as predicted by our theory.

\begin{table}[htbp!]
\begin{center}%
\begin{scriptsize}
\begin{tabular}
[c]{|c|c|c|c|c|}\hline
$\gamma=1/6$ & $(\widehat{h}\sub, \widehat{sd(h\sub)})$ &$(\widehat{\theta}_{2}, \widehat{sd(\theta_{2})})$&$(\widehat{\theta}_{3}, \widehat{sd(\theta_{3})})$ &$\widehat{CT}\sub$   \\
\hline  \text{standard MALA}& $(0.537, 0.043)$& $(1.071, 0.022)$& $(-0.0002, 0.062)$& $0.725$\\
        $\alpha=1.2$ & $(0.464,0.039)$ &  $(1.064,0.025)$   & $(0.004, 0.071)$   & $1.289$      \\
        $\alpha=1.5$ & $(0.472,0.041)$ &  $(1.063,0.025)$   & $(-0.0016, 0.073)$   & $1.175$    \\
        $\alpha=1.8$ & $(0.481, 0.041)$&  $(1.061,0.025)$   & $(-0.0015, 0.074)$   & $1.070$     \\
        $\alpha=2.0$ & $(0.488,0.040)$ &  $(1.061, 0.025)$  & $(0.0005, 0.077)$   & $1.005$     \\
        $\alpha=2.5$ & $(0.507, 0.043)$ &  $(1.059, 0.024)$ & $(-0.0041,0.076)$   & $1.038$      \\
        $\alpha=3.0$ & $(0.535, 0.044)$&  $(1.061, 0.023)$   & $(0.00068, 0.077)$   & $1.063$    \\
\hline
$\gamma=1/2$ &  & &  &   \\
\hline  \text{standard MALA}& $(0.753, 0.008)$& $(1.071, 0.020)$& $(-0.0006, 0.066)$& $0.427$\\
        $\alpha=1.2$ & $(0.646,0.008)$ &  $(1.065,0.024)$   & $(-0.0025, 0.082)$   & $0.663$      \\
        $\alpha=1.5$ & $(0.649,0.008)$ &  $(1.066,0.024)$   & $(0.0025, 0.082)$   & $0.466$    \\
        $\alpha=1.8$ & $(0.653, 0.009)$&  $(1.063,0.024)$   & $(-0.0011, 0.081)$   & $0.327$     \\
        $\alpha=2.0$ & $(0.655,0.009)$ &  $(1.064, 0.024)$  & $(0.0010, 0.083)$   & $0.556$     \\
        $\alpha=2.5$ & $(0.662, 0.011)$ &  $(1.063, 0.023)$ & $(0.0011,0.084)$   & $0.549$      \\
        $\alpha=3.0$ & $(0.668, 0.012)$&  $(1.064, 0.023)$   & $(0.00317, 0.078)$   & $0.541$    \\\hline
\end{tabular}
\end{scriptsize}
\end{center}
\caption{Dimension $N=10$ and $S^{N}=S^{N}_{2}$. }
\label{Table_d10}%
\end{table}

\begin{table}[htbp!]
\begin{center}%
\begin{scriptsize}
\begin{tabular}
[c]{|c|c|c|c|c|}\hline
$\gamma=1/6$ & $(\widehat{h}\sub, \widehat{sd(h\sub)})$ &$(\widehat{\theta}_{2}, \widehat{sd(\theta_{2})})$&$(\widehat{\theta}_{3}, \widehat{sd(\theta_{3})})$ &$\widehat{CT}\sub$   \\
\hline  \text{standard MALA}& $(0.538, 0.052)$& $(1.054, 0.025)$& $(-0.0035, 0.086)$& $0.900$\\
        $\alpha=1.2$ & $(0.463,0.047)$ &  $(1.046,0.028)$   & $(-0.0023, 0.100)$   & $1.826$     \\
        $\alpha=1.5$ & $(0.476,0.048)$ &  $(1.046,0.027)$   & $(-0.0021, 0.097)$   & $1.485$    \\
        $\alpha=1.8$ & $(0.489, 0.049)$&  $(1.046,0.027)$   & $(0.0016, 0.096)$   & $1.209$     \\
        $\alpha=2.0$ & $(0.496,0.052)$ &  $(1.046, 0.031)$  & $(0.0017, 0.094)$   & $1.053$     \\
        $\alpha=2.5$ & $(0.518, 0.051)$&  $(1.046, 0.026)$ & $(-0.0028,0.096)$   & $1.037$      \\
        $\alpha=3.0$ & $(0.528, 0.053)$&  $(1.047, 0.028)$   & $(0.00031, 0.093)$   & $1.021$    \\
\hline
$\gamma=1/2$ &  & &  &   \\
\hline  \text{standard MALA}& $(0.949, 0.003)$& $(1.036, 0.038)$& $(-0.0015, 0.144)$& $0.776$\\
        $\alpha=1.2$ & $(0.788,0.006)$ &  $(1.021,0.043)$   & $(-0.0007, 0.169)$   & $1.326$      \\
        $\alpha=1.5$ & $(0.832,0.005)$ &  $(1.023,0.041)$   & $(0.0008, 0.166)$   & $0.684$    \\
        $\alpha=1.8$ & $(0.865, 0.005)$&  $(1.028,0.041)$   & $(-0.0043, 0.158)$   & $0.361$     \\
        $\alpha=2.0$ & $(0.883,0.005)$ &  $(1.030, 0.041)$  & $(-0.0021, 0.158)$   & $0.873$     \\
        $\alpha=2.5$ & $(0.914, 0.004)$ &  $(1.031, 0.038)$ & $(-0.0001,0.153)$   & $0.828$      \\
        $\alpha=3.0$ & $(0.931, 0.003)$&  $(1.034, 0.037)$   & $(0.0002, 0.151)$   & $0.803$    \\\hline
\end{tabular}
\end{scriptsize}
\end{center}
\caption{Dimension $N=50$ and $S^{N}=S^{N}_{2}$.}
\label{Table_d50}%
\end{table}

\begin{table}[htbp!]
\begin{center}%
\begin{scriptsize}
\begin{tabular}
[c]{|c|c|c|c|c|}\hline
$\gamma=1/6$ & $(\widehat{h}\sub, \widehat{sd(h\sub)})$ &$(\widehat{\theta}_{2}, \widehat{sd(\theta_{2})})$&$(\widehat{\theta}_{3}, \widehat{sd(\theta_{3})})$ &$\widehat{CT}\sub$   \\
\hline  \text{standard MALA}& $(0.549, 0.054)$& $(1.045, 0.027)$& $(0.0015, 0.102)$& $0.955$\\
        $\alpha=1.2$ & $(0.473,0.049)$ &  $(1.037,0.029)$   & $(-0.0012, 0.115)$   & $2.080$     \\
        $\alpha=1.5$ & $(0.484,0.054)$ &  $(1.036,0.034)$   & $(-0.0028, 0.111)$   & $1.623$    \\
        $\alpha=1.8$ & $(0.498, 0.054)$&  $(1.037,0.036)$   & $(-0.0021, 0.109)$   & $1.268$     \\
        $\alpha=2.0$ & $(0.507,0.056)$ &  $(1.041, 0.084)$  & $(0.0074, 0.235)$   & $1.076$     \\
        $\alpha=2.5$ & $(0.523, 0.054)$&  $(1.041, 0.027)$  & $(-0.0018,0.106)$   & $1.052$      \\
        $\alpha=3.0$ & $(0.533, 0.057)$&  $(1.039, 0.034)$  & $(0.00019, 0.104)$   & $1.031$    \\
\hline
$\gamma=1/2$ &  & &  &   \\
\hline  \text{standard MALA}& $(0.975, 0.002)$& $(1.006, 0.050)$& $(0.0103, 0.206)$& $0.970$\\
        $\alpha=1.2$ & $(0.809,0.006)$ &  $(0.981,0.055)$  & $(0.0039, 0.228)$   & $1.705$      \\
        $\alpha=1.5$ & $(0.865,0.005)$ &  $(0.989,0.053)$   & $(0.0138, 0.216)$   & $0.785$    \\
        $\alpha=1.8$ & $(0.904, 0.004)$&  $(0.993,0.054)$   & $(-0.0066, 0.213)$   & $0.371$     \\
        $\alpha=2.0$ & $(0.922,0.004)$ &  $(0.997, 0.049)$  & $(0.0004, 0.215)$   & $1.055$     \\
        $\alpha=2.5$ & $(0.952, 0.003)$ &  $(1.002, 0.049)$ & $(0.0018,0.211)$   & $1.009$      \\
        $\alpha=3.0$ & $(0.966, 0.002)$&  $(1.002, 0.049)$   & $(-0.0016, 0.204)$   & $0.985$    \\\hline
\end{tabular}
\end{scriptsize}
\end{center}
\caption{Dimension $N=100$ and $S^{N}=S^{N}_{2}$.}
\label{Table_d100}%
\end{table}

To visualize the situation, in Figure \ref{figsim1} we have plotted $\alpha$ versus the statistic $\widehat{CT}\sub$. In each case, standard MALA is represented by the rightmost value for $\alpha$. As we observe from Tables \ref{Table_d10}-\ref{Table_d100} and from Figure \ref{figsim1} and also indicated by the theory, when $\gamma>1/6$ (in particular $\gamma=1/2$ here), $\widehat{CT}\sub$ is smaller when $1<\alpha<2$. In particular, notice that for $\gamma=1/2>1/6$, $\widehat{CT}\sub$  seems to obtain a minimum value around $\alpha \approx 1.8<2$, which is in the fluid regime.
\begin{figure}[!h]
 \centering
\begin{tabular}{cc}
 \includegraphics[scale=0.3]{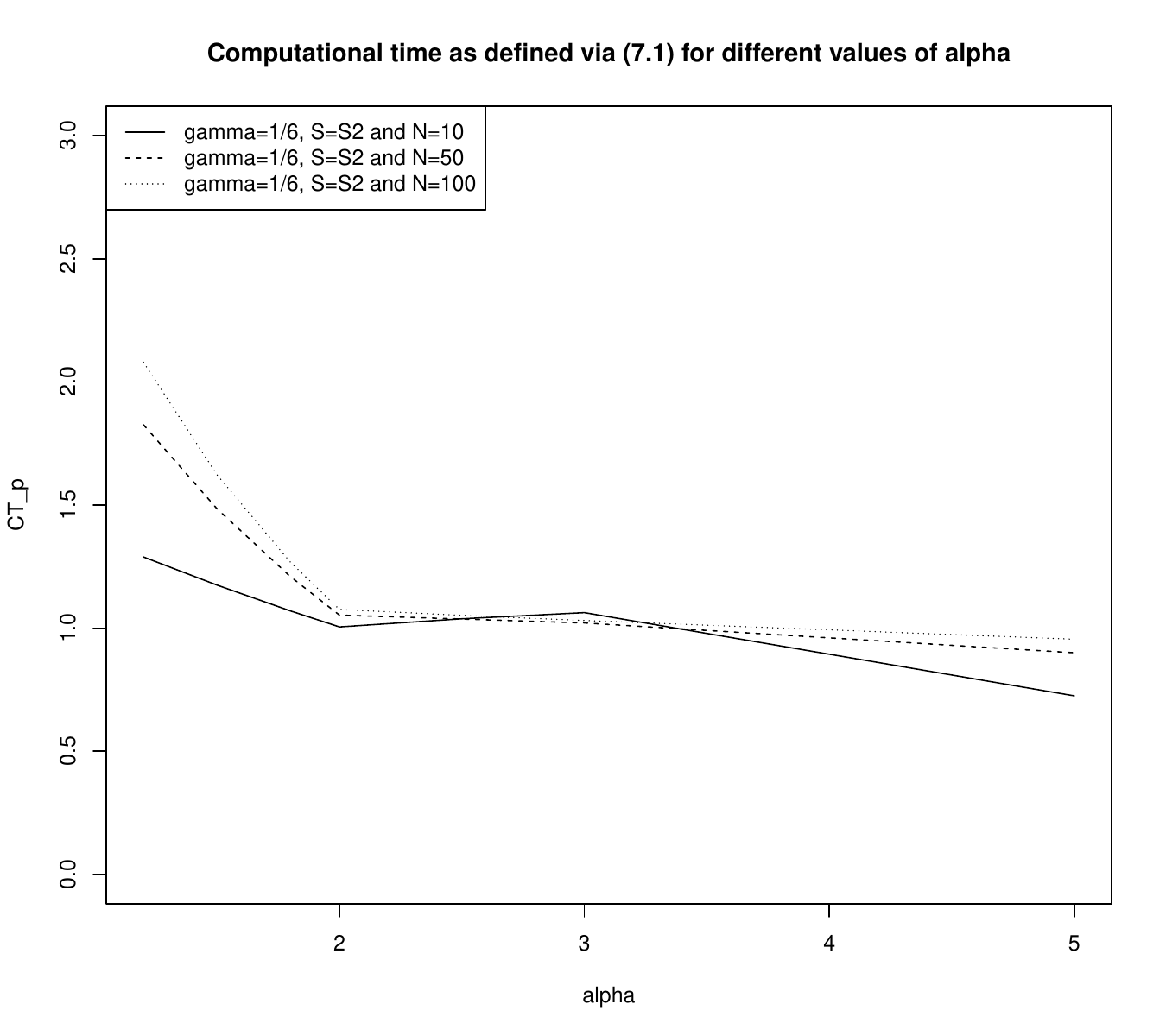}
  &
  \includegraphics[scale=0.3]{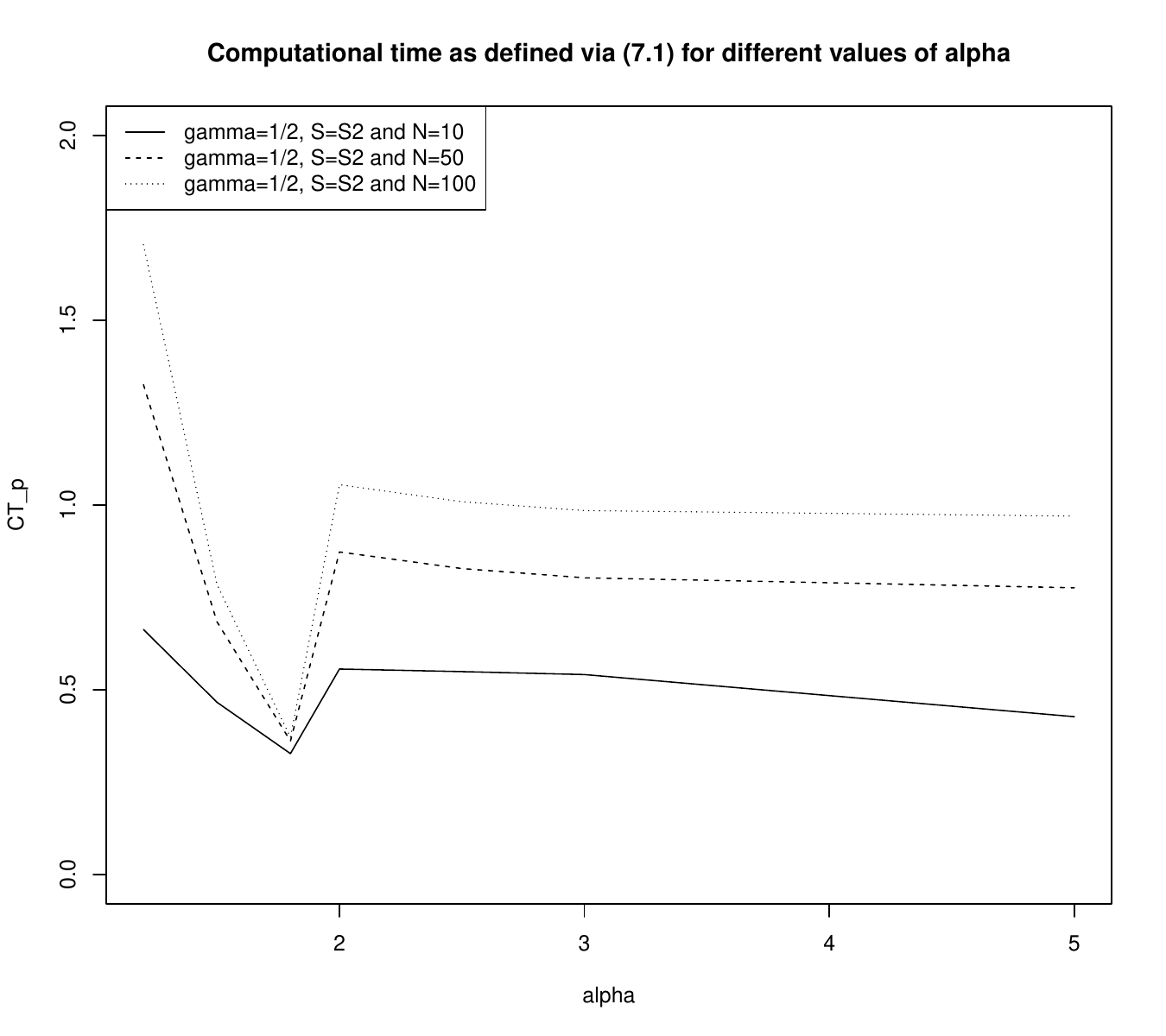}\\
 (a) $S^{N}=S^{N}_{2}$  and $\gamma=1/6$ &  (b) $S^{N}=S^{N}_{2}$  and $\gamma=1/2$\\
 \end{tabular}
\caption{Measure $\widehat{CT}\sub$ versus values of $\alpha$. In each of the figures the rightmost value for $\alpha$ corresponds to standard MALA} \label{figsim1}
 \end{figure}

The numerical studies presented above illustrate the theoretical findings of the paper and also show that the choice of the irreversible perturbations $S$ can have serious consequences on the performance of the algorithm.

In particular, in the case of $S=S^{N}_{1}$
\begin{itemize}
\item{the theoretical acceptance probability for \imala\ is increasing as a function of $\alpha$. This means that by increasing $\alpha$ one can increase the optimal acceptance probability considerably and combined with the fact that (a) the limit of the interpolation of the chain is a diffusion as in the standard MALA, and that (b) for large values of $\alpha$ the values of $\widehat{CT}\sub$ are the same for \imala\ and standard MALA, we have an algorithm with the same cost but much higher acceptance probabilities.}
 \item{the statistic $\widehat{CT}\sub$ for \imala\ converges to that  for the standard MALA case as $\alpha$ increases. However, when $1<\alpha<2$ the performance seems to be worse than standard MALA due to very low optimal acceptance probabilities in this case.}
 \item{for large $N$, the estimated standard deviation of the statistics that were computed, $\theta_{2}$ and $\theta_{3}$, appear to be slightly larger than  the corresponding estimated standard deviation of the statistics for the standard MALA case, but certainly of the same order. }
 \end{itemize}
On the other hand, in the case of $S=S^{N}_{2}$
\begin{itemize}
\item{ the theoretical acceptance probability for \imala\ is the same as that of standard MALA, i.e., around $0.574$, even though this seems to be realizable in practice in very high dimensions as even for $N=100$, the estimated optimal acceptance probabilities were lower.}
 \item{the statistic $\widehat{CT}\sub$ has lower values for $\gamma=1/2>1/6$ and $1<\alpha<2$ when compared to the corresponding values for the standard MALA case. This can be explained by the cost of the algorithm which in that case is $N^{\alpha\gamma}$ with $\alpha\gamma<1/3$ as opposed to the $N^{1/3}$ which is the cost of the standard MALA case. As $\alpha$ increases, the value of $\widehat{CT}\sub$ seems to approach that of standard MALA.}
 \item{the estimated standard deviation of the statistics that were computed, $\theta_{2}$ and $\theta_{3}$, are either the same or slightly larger than the corresponding estimated standard deviation of the statistics for the standard MALA case, but again of the same order. }
 \end{itemize}

 \section{Acknowledgments}\label{S:Acknowledgments}
 K.S. was partially supported by NSF CAREER AWARD NSF-DMS 1550918. 

\appendix

\section{Proof of Theorem \ref{mainthm}}\label{AppendixA}

In this section we present the proof of our main results. The proof is based on diffusion approximation techniques analogous to those used in \cite{MattinglyPillaiStuart2011}. In \cite{MattinglyPillaiStuart2011} the authors consider the MALA algorithm with reversible proposal. That is, if we fix $S=0$ in our paper and $\Psi=0$ in their paper, the algorithms we consider coincide.  For this reason we try to adopt a notation as similar as possible to the one used in \cite{MattinglyPillaiStuart2011} and, for the sake of brevity,  we detail only the parts of the proof that differ from the work \cite{MattinglyPillaiStuart2011} and just sketch the rest.

We start by  recalling that, by Fernique's theorem,
\be\label{finitemoments}
\Epin \| x^N\|^p =\EE \|\cCn^{1/2} z^N\|^p \less 1, \quad \mbox{for all } p\geq 1.
\ee
This fact will be often implicitly used without mention in the remainder of the paper.

We also recall that the chain $\{x_k^N\}_k$ that we consider is defined in \eqref{chain};  the drift-martingale decomposition of such a  chain is given in equation \eqref{dr-martal2} . 
Let us start by recalling the definition of  the continuous interpolant of the chain, equations \eqref{zeta=zetaalpha}-\eqref{continter},  and by introducing the piecewise constant interpolant of the chain $x_k^N$, that is,
\begin{equation*}
\bar{x}\bn(t)=x_{k}^N \qquad t_k\leq t< t_{k+1},
\end{equation*}
 where  $ t_k=k/N^{\zeta\gamma}$. It is easy to see (see e.g. \cite[Appendix A]{OPPS}) that \be\label{wf}
x\bn(t)=x_0^N+\int_{0}^t d\sub^N(\bar{x}\bn (v)) dv+   w\sub^N(t),
\ee
where
$$
w^N\sub(t):= \frac{1}{N^{\zeta\gamma/2}} \sum_{j=1}^{k-1}M_j^N +N^{\zeta\gamma/2}(t-t_k)M_k^N.
$$
For any  $t \in [0,T]$, we set
\begin{align}
\hat{w}\sub^N(t)&:= \int_0^t \left[ d\sub^N(\bar{x}\bn (v))-d\sub(x\bn(v)) \right] dv + w\sub^N(t)\nonumber \\
&= \int_0^t \left[ d\sub^N(\bar{x}\bn (v))-
d\sub(\bar{x}\bn(v)) \right] dv \nonumber\\
&+\int_0^t \left[ d\sub(\bar{x}\bn(v)) - d\sub({x}\bn(v)) \right] dv +w\sub^N(t). \label{whatN}
\end{align}
With the above notation, we can then rewrite \eqref{wf} as
\begin{align}\label{contmapetahatx}
x\bn(t)&=x_0^N+\int_{0}^t d\sub(x\bn(v)) dv+\hat{w}\sub^N(t).
\end{align}
Let now $C([0,T];\mathcal{H})$ denote the space of $\mathcal{H}$-valued continuous functions, endowed with the uniform topology and   consider  the map
\begin{align*}
\mathcal{I}:  \mathcal{H} \times C([0,T];\mathcal{H}) & \longrightarrow C([0,T];\mathcal{H})\\
 (x_0 , \eta(t)) & \longrightarrow  x(t).
\end{align*}
That is, $\mathcal{I}$ is the map that to every $(x_0 , \eta(t)) \in \mathcal{H} \times C([0,T];\mathcal{H})$ associates the (unique solution) of the equation
\begin{align}
x(t)=x_0+\int_0^t d\sub(x(s)) ds + \eta(t).\label{Eq:LimitProcess}
\end{align}
From \eqref{contmapetahatx} it is clear that $x\bn= \mathcal{I}(x_0^N, \hat{w}\sub^N)$. Notice that, under our continuity assumption on $\tS$, $\mathcal{I}$ is a continuous map. Therefore, in order to prove that $x\bn(t)$ converges weakly to $x(t)$ (where $x(t)$ is the solution of equation \eqref{Eq:LimitProcess} with $\eta(t)=D\sub W^{\cC}(t)$), by the continuous mapping theorem we only need to prove that   $\hat{w}\sub^N$ converges weakly to $D\sub W^{\cC}(t)$, where $W^{\cC}(t)$ is a $\mathcal{H}$-valued $\cC$-Brownian motion.
The weak convergence of $\hat{w}\sub^N$  to $D\sub W^{\cC}(t)$   is a consequence of \eqref{whatN} and of  Lemmata \ref{lemma1} and \ref{lemma2}. Then, we get the statement of Theorem \ref{mainthm}. The proof of Lemma \ref{lemma1} and Lemma \ref{lemma2} is contained in the remained of  this Appendix.

\begin{lemma}\label{lemma1}
Under Assumption \ref{AssCLT}, Assumption \ref{Ass2} and Condition \ref{Ass1}, the following holds
\be \label{l1}
\Epin\int_0^T \| d\sub^N(\bar{x}\bn (t))-
d\sub(\bar{x}\bn(t)) \| dt \longrightarrow 0, \quad \mbox{as } N \ra \infty
\ee
and
\be\label{l2}
\Epin \int_0^T \| d\sub (\bar{x}^N(t)) - d\sub ({x}\bn(t)) \| dt \longrightarrow 0, \quad \mbox{as } N \ra \infty,
\ee
where  the function $d\sub (x)$ has been defined in the statement of Theorem \ref{mainthm}.
\end{lemma}
Set now
\begin{equation}\label{epsx}
\epsilon\sub^N(x):= \EE_x \left[\left(1 \wedge e^{Q^N} \right) \cC_N^{1/2}z^N \right]
\end{equation}
and
\begin{equation*}
h^N\sub(x):= \EE_x \left(1 \wedge e^{Q^N} \right).
\end{equation*}
While $h\sub$ (see \eqref{hS}) is the limiting average acceptance probability, $h\sub^N(x)$ is the {\em local} average acceptance probability. The above notation will be used in the proof of the next lemma.

\begin{lemma}\label{lemma2}
If Assumption \ref{AssCLT}, Assumption \ref{Ass2} and Condition \ref{Ass1} hold, then $w\sub^N(t)$ converges weakly in $C([0,T];\cH)$ to
$D\sub W^{\cC}(t)$, where $W^{\cC}(t)$ is a $\mathcal{H}$-valued $\cC$-Brownian motion and the constant $D\sub$ has been defined in the statement of Theorem \ref{mainthm}.
\end{lemma}

\begin{proof}[Proof of Lemma \ref{lemma1}]
We start by proving \eqref{l2}, which is simpler. The drift coefficient $d\sub$ is globally Lipshitz; therefore, using  \eqref{chain},  \eqref{propprodgas} and \eqref{defsigma}, if $t_k\leq t <t_{k+1}$, we have
\begin{align*}
\Epin \| d\sub (\bar{x}\bn (t)) - d\sub ({x}\bn(t)) \|&\less \Epin \| \bar{x}\bn(t) - x\bn(t)\| \\
& \less \lv (N^{\zeta\gamma} t - k)   \rv \Epin \| x_{k+1}^N-x\kn \| \less \Epin \| y_{k+1}^N-x\kn \| \\
& \less \left(\frac{1}{N^{2\gamma}}+\frac{1}{N^{\alpha\gamma}} \right)
\Epin \| x_k^N\| + \frac{1}{N^{\gamma}}
\mathbb{E} \| (\cC^N)^{1/2}z_{k+1}^N\| \ra 0.
\end{align*}

Let us now come to the proof of \eqref{l1}.
From \eqref{apprdrift1}-\eqref{apprdrift2}, we have
$$
d^N\sub(x)- d\sub(x) =A_1^N+A_2^N+A_3^N - d\sub(x),
$$
where
\begin{align}
A_1^N & :=  N^{\zeta\gamma}\EE_x\left[ (1 \wedge e^{Q^N}) \left(- \frac{\ell^2}{2N^{2\gamma}}  x^N\right) \right]  \label{A1}\\
A_2^N & :=  N^{(\zeta-\alpha)\gamma} \ell^{\alpha}\EE_x\left[ (1 \wedge e^{Q^N}) \tilde{S}^N x^N\right] \nonumber\\
A_3^N & := N^{(\zeta-1)\gamma} \ell  \EE_x\left[ (1 \wedge e^{Q^N}) \cC_N^{1/2} z^N \right]
\stackrel{\eqref{epsx}}{=} N^{(\zeta-1)\gamma} \ell \epsilon\sub^N(x) . \label{A3333}
\end{align}
We split the function $d\sub(x)$ in three (corresponding) parts:
$$
d\sub(x) = A_1(x)+A_2(x)+A_3(x),
$$
with
$$
A_1:= \left\{
\begin{array}{ll}
- \frac{\ell^2}{2} h\sub x  & \mbox{if } \gamma = 1/6 \mbox{ and } \alpha \geq 2\\
0 & \mbox{otherwise}
\end{array}
\right.
$$
$$
A_2:= \left\{
\begin{array}{ll}
h\sub {\ell^{\alpha}} \tilde{S} x  & \mbox{if } \gamma \geq 1/6 \mbox{ and } \alpha  \leq 2\\
0 & \mbox{otherwise}
\end{array}
\right.
$$
$$
A_3:= \left\{
\begin{array}{ll}
-2 \nu\sub {\ell^{\alpha}} \tilde{S} x  & \mbox{if } \gamma \geq 1/6 \mbox{ and } 1 \leq  \alpha  \leq 2\\
0 & \mbox{otherwise}
\end{array}
\right.
$$

We therefore need to consecutively estimate the above three terms.

\smallskip

$ \bullet {\bf \,\, A_1^N-A_1}:$ if $\alpha \geq 2$ (and $\gamma=1/6$) we fix $\zeta=2$ and we have
\begin{align}
\Epin \| A_1^N-A_1\|  & \less \Epin \left[ \lv \EE_x \left(1 \wedge e^{Q^N} \right) -h\sub \rv \| x\|  \right]+ \Epin \nor{x^N-x} \nonumber\\
& \stackrel{\eqref{finitemoments}}{\less}
\left(\Epin \lv h^N\sub(x) -h\sub \rv ^2\right)^{1/2}+ \Epin \nor{x^N-x}\longrightarrow  0, \label{sss1}
\end{align}
as the first addend tends to zero by Lemma \ref{lemmapreliminaries} of \ref{AppendixB} and the second addend tends to zero by definition (see also \cite[equation (4.3)]{PST}). If $\alpha<2$ then $\zeta=\alpha$ and we have
\be\label{sss}
\Epin \|A_1^N\|\less N^{(\alpha-2)\gamma}
\Epin\| \EE_x (1 \wedge e^{Q^N}) x^N \| \less N^{(\alpha-2)\gamma} \ra 0\,.
\ee

$ \bullet {\bf \,\, A_2^N-A_2}:$ if $\alpha \leq 2$ then, recalling \eqref{zeta=zetaalpha},  a calculation analogous to the one in \eqref{sss1} gives the statement. If $\alpha>2$ then we can act as in \eqref{sss}.

$ \bullet {\bf \,\, A_3^N-A_3}:$ by Lemma \ref{lemmapreliminaries} of \ref{AppendixB} (and \eqref{zeta=zetaalpha}), we have
$$
\Epin \|A_3^N - A_3\| \ra 0 \quad \mbox{as } N \ra \infty.
$$
This concludes the proof.
\end{proof}
\begin{proof}[Proof of Lemma \ref{lemma2}] The calculations here are standard so we only prove it for the case $\gamma=1/6$ and $\alpha>2$.
Let us recall the martingale difference given by (\ref{Eq:GeneralMartingaleDiff}).
In the case $\gamma=1/6$ and $\alpha>2$, we have $\zeta=2$ and $d\sub (x)=-\frac{\ell^2}{2}h\sub x$. Hence, the expression (\ref{Eq:GeneralMartingaleDiff}) becomes
\begin{align}
M_k^N&=N^{-1/6}\left(\tilde{\beta}^{N}\frac{1}{h\sub}d\sub (x_k^N)\right) + N^{(1-\alpha)/6}\left(\ell^{\alpha}\tilde{\beta}^{N} \tilde{S} x_k^N\right)+
\left(\ell\tilde{\beta}^{N} \cC^{1/2}z_{k+1}^N\right)- N^{-1/6}d\sub^{N}(x_k^N)\nonumber\\
&=\left(\ell\tilde{\beta}^{N} \cC^{1/2}z_{k+1}^N\right)+ N^{-(\alpha-1)/6}\left(\ell^{\alpha}\tilde{\beta}^{N} \tilde{S} x_k^N\right) -  N^{-1/6}\frac{1}{h\sub}\left[h\sub d\sub^{N}(x_k^N)-\tilde{\beta}^{N}d\sub (x_k^N)\right]\nonumber
\end{align}

By Lemma \ref{lemma1}, we have that $\Epin  \| d\sub^N(x)-d\sub (x) \|^{2}  \longrightarrow 0$ as $N\rightarrow\infty$. This implies that
\begin{align*}
 N^{-1/3}\Epin \| h\sub d\sub^{N}(x)-\tilde{\beta}^{N}d\sub (x)\|^{2}&\rightarrow 0.
\end{align*}

At the same time, we also notice that
\begin{align*}
N^{-(\alpha-1)/3}\ell^{2\alpha}\Epin \|\tilde{\beta}^{N} \tilde{S} x\|^2
& \less  N^{-(\alpha-1)/3}  \ell^{2\alpha} \Epin \|  x\|^2 \longrightarrow 0  \quad \mbox{if } \alpha>1.
\end{align*}

Hence, if we define $\mathcal{M}_{N}(x)=\EE_x\left[M_k^N \otimes M_k^N|x_k^N=x\right]$, we obtain that up to a constant
\begin{align*}
 \Epin \left|\text{Tr}(\mathcal{M}_{N}(x))-\EE_x\left[\|\ell\tilde{\beta}^{N} \cC^{1/2}z\|^2\right]\right|&\leq N^{-1/3}
\end{align*}

Then, as in Lemma 4.8 of \cite{PST} we obtain that
\begin{align*}
 \Epin \left|\ell^{2}h\sub\text{Tr}(\cC)-\EE_x\left[\|\ell\tilde{\beta}^{N} \cC^{1/2}z\|^2\right]\right|&\rightarrow 0, \text{ as } N\rightarrow\infty
\end{align*}
which then immediately implies that
\begin{align*}
 \Epin \left|\ell^{2}h\sub\text{Tr}(\cC)-\text{Tr}(\mathcal{M}_{N}(x))\right|&\rightarrow 0, \text{ as } N\rightarrow\infty.
\end{align*}
The latter result implies that the invariance principle of Proposition 5.1 of \cite{Berger} holds, which then imply the statement of the lemma.
\end{proof}
\section{Auxiliary estimates}\label{AppendixB}

We first decompose $Q^N$ 
as follows: let
\begin{align}\label{tildeQN}
{Z}^N := & - \frac{\ell^6}{32}-a - \frac{\ell^3}{4 {N^{3\gamma}}}
 \lanc{x}{\cC^{1/2} z^N}\nonumber \\
& -2 \frac{\ell^{\alpha-1}}{N^{(\alpha-1)\gamma}} \langle (\cC^N)^{1/2}z^N, S^N x^N \rangle
- \frac{\ell^{(2\alpha-1)}}{N^{(2\alpha-1)\gamma}} \lancn{\tS^N (\cC^N)^{1/2} z^N}{\tS^N x^N} \nonumber\\
& - \frac{\ell^{(3\alpha-1)}}{N^{(3\alpha-1)\gamma}} \lancn{\tS^N \cCn^{1/2} z^N}{(\tS^N)^2 x^N} \,.
\end{align}
Then
\be\label{q=Q^N+estar}
Q^N = {Z}^N + e_{\star}^N,
\ee
where
\begin{align}\label{estar}
e_{\star}^N  & :=   \frac{\ell^6}{32}- \frac{\ell^6}{32 \,N^{6\gamma}} \norcn{x^N}^2 \nonumber \\
 &+ a -2 \frac{\ell^{2(\alpha-1)}}{N^{(\alpha-1) 2\gamma}}\norcn{\tS^N x^N}^2
- \frac{\ell^{2(2\alpha-1)}}{2 N^{2 \gamma (2\alpha-1)}}\norcn{\tSn^2 x^N}^2 \nonumber \\
& + i^N(x,z)+ e^N(x,z) ,
\end{align}
with
\begin{equation*}
i^N(x,z):= i_1^N(x,z) + i_2^N(x,z),
\end{equation*}
having defined
\begin{align}
& i_1^N(x,z): =\frac{\ell^4}{8 N^{4\gamma}} \left( \norcn{x^N}^2- \| z^N\|^2\right) \label{defi1}\\
& i_2^N(x,z): =\frac{\ell^{2\alpha}}{2 N^{2 \alpha \gamma}} \left( \norcn{\tS^N x^N}^2-
\norcn{\tS^N \cCn^{1/2}z^N}^2\right), \label{defi2}
\end{align}
and
\begin{align}\label{defeNN}
e^N &:=  \frac{\ell^5}{8 N^{5\gamma}} \lancn{x^N}{\cCn^{1/2}z^N}-
\frac{\ell^{2(\alpha+1)}}{4 N^{2(1+\alpha)\gamma}}
 \norc{\tS^N x^N}^2\nonumber \\
& - \frac{\ell^{3+\alpha}}{4 N^{(3+\alpha)\gamma}} \langle \cCn^{1/2}z^N, S^N x^N \rangle
+ \frac{\ell^{(2\alpha+1)}}{2 N^{(2\alpha+1)\gamma}} \lancn{\tS \cCn^{1/2} z^N}{\tS^N x^N}.
\end{align}
Finally, we set
\be\label{tildeeN}
\tilde{e}^N(x,z):= e^N+\frac{\ell^{2(\alpha+1)}}{4 N^{2(1+\alpha)\gamma}}
 \norcn{\tS^N x^N}^2.
\ee
That is, $\tilde{e}^N$ contains only the addends of $e^N$ that depend on the noise $z$.

Furthermore, we split $Q^N(x,z)$ into the terms that contain $z^{j,N}$ and the terms that don't, $Q^N_j$ and $Q^N_{j, \perp}$, respectively; that is
$$
Q^N= Q^{N}_j+ Q^N_{j, \perp},
$$
where
\be\label{defQj}
Q_j^N:= \tilde{e}^N+ (i_1^N)_j+ (i_2^N)_j+ {Z}^N_j,
\ee
having denoted by $(i_1^N)_j, \, (i_2^N)_j$ and  ${Z}^N_j$, the part of $i_1^N, i_2^N$ and
${Z}^N$, respectively, that depend on $z^{j,N}$.

\begin{lemma}\label{variousestimates} Let Assumption \ref{AssCLT}, \,Assumption \ref{Ass2}\, and Condition \ref{Ass1} hold;  then,
\begin{align}
\textrm{i)}& \quad  \Epin \lv \tilde{e}^N \rv^2 \less \frac{1}{N^{2/3}}
+ \frac{1}{N^{4 \gamma}}, \quad
 \mbox{ for all } \alpha\geq 1, \, \gamma \geq 1/6 \label{lem911}\\
\textrm{ii)}& \quad  N^{1/3} \Epin \sum_{j=1}^N
\lambda_j^2 \lv (i_1^N)_j\rv^2  \longrightarrow 0, \quad   \mbox{as }\, N\ra \infty\nonumber
\\
\textrm{iii)}& \quad  \Epin \lv i_2^N\rv^2 \less \frac{1}{N^{4 \gamma}}, \mbox{ for all } \alpha\geq 1, \gamma\geq 1/6\nonumber
\\
\textrm{iv)}& \quad  N^{1/3} \Epin \sum_{j=1}^N
\lambda_j^2 \lv {Z}^N_j\rv^2  \longrightarrow 0, \quad \mbox{as }\,N \ra \infty,  \mbox{ for all } \alpha>2 , \gamma=1/6
\label{lem914}\\
\textrm{v)}& \quad \Epin \lv e^N \rv^2 \less \frac{1}{N^{2/3}}
+ \frac{1}{N^{4 \gamma}}, \quad
\mbox{ for all } \alpha\geq 1, \, \gamma \geq 1/6
\\
\textrm{vi)}& \quad \Epin \lv (\mathrm{Var}_x {Z}^N)^{1/2} - (\mathrm{Var} Q)^{1/2} \rv^2 \longrightarrow 0, \quad  \mbox{as }\, N \ra \infty, \mbox{ for all } \alpha\geq 1, \, \gamma \geq 1/6
\,.\label{ve6}
\end{align}
\end{lemma}
\begin{proof}[Proof of Lemma \ref{variousestimates}]Recall that, under $\pi^N$, $x^{i, N} \sim \lambda_i \rho^i$, where $\{\rho^i\}_{i \in \N}$ are i.i.d standard Gaussians.  We now consecutively prove all  the  statements of the lemma. \\
{\bf Proof of i).}  Notice that
$$
\Epin \lv \lancn{x^N}{\cCn^{1/2}z^N} \rv^2 = \EE \lv \sum_{i=1}^N \rho^i z^{i,N}\rv^2 =N \,.
$$
Therefore, since $\gamma \geq 1/6$,
\be\label{lve1}
\Epin \lv \frac{\ell^5}{8 N^{5\gamma}} \lancn{x^N}{\cCn^{1/2}z^N}\rv^2 \less N^{-2/3}.
\ee
Furthermore, using \eqref{f1} (which follows from point i) of Assumption \ref{Ass2}) and \eqref{constc1=},  we have
\be\label{lve2}
 \frac{1}{N^{(3+\alpha)2 \gamma}} \Epin \lv \langle \cCn^{1/2}z^N, S^Nx^N \rangle  \rv^2
\less N^{- 8 \gamma};
\ee
similarly,   using \eqref{f2} (which follows from point ii) of Assumption \ref{Ass2}) and \eqref{constc2=},
\be\label{lve3}
  \frac{1}{ N^{(2\alpha+1)2 \gamma}}
 \Epin \lv  \lanc{\tS \cCn^{1/2} z^N}{\tS^N x^N}\rv^2 \less N^{-4 \gamma} \,.
\ee
Now the first statement of the lemma is a consequence of \eqref{defeNN}-\eqref{tildeeN} and the above \eqref{lve1}, \eqref{lve2} and \eqref{lve3}. \\
{\bf Proof of ii).} This is proved in \cite{PST}, see calculations after \cite[equation (4.18)]{PST}, so we omit it. \\
{\bf Proof of iii).} This estimate follows again from Assumption \ref{Ass2}, once we observe that if $x \sim \pi^N$ then
$ \norcn{ \tS^N x^N}^2$ and  $\norcn{\tS^N \cCn^{1/2} z^N}^2$ are two independent random variables with the same distribution.
With this observation in place, we have
\begin{align*}
\Epin \lv \frac{ \norc{ \tS x^N}^2 - \norc{\tS \cC_N^{1/2} z^N}^2}{N^{2 \alpha \gamma}} \rv^2 & =
\Epin \lv \frac{ \norc{ \tS x^N}^2 }{N^{2\alpha \gamma}} - \frac{c_1}{N^{2 \gamma}}
+ \frac{c_1}{N^{2 \gamma}}
-   \frac{\norc{\tS \cC_N^{1/2} z^N}^2}{N^{2\alpha \gamma}}   \rv^2 \\
& \less \Epin \lv \frac{ \norc{ \tS x^N}^2 }{N^{2\alpha \gamma}} - \frac{c_1}{N^{2 \gamma}} \rv^2 = N^{-4\gamma}  \Epin \lv \frac{ \norc{ \tS x^N}^2 }{N^{2(\alpha-1)\gamma}} - c_1 \rv^2,
\end{align*}
which gives the claim. \\

{\bf Proof of iv).} We recall that ${Z}^N_j$ has been introduced in \eqref{defQj}. Using the antisymmetry of $S^N$ and the definition of $\tS^N$, we have
$$
-\lanc{\tS^N \cCn^{1/2}z^N}{\tS^N x^N}= \langle  \cCn^{1/2}z^N,  S^N \tS^N x^N\rangle
$$
and
$$
-\lancn{\tS^N \cCn^{1/2}z^N}{\tSn^2 x}=  \langle  \cCn^{1/2}z^N,  S^N \tSn^2 x^N\rangle .
$$
We can therefore write an explicit expression for ${Z}^N_j$:
\begin{align}
{Z}^N_j & = -\frac{\ell^3}{4 {N^{3\gamma}}} \frac{x^{j,N} z^{j,N}}{\lambda_j }
-2 \frac{\ell^{\alpha-1}}{N^{(\alpha-1)\gamma}} \lambda_j z^{j,N} (S^Nx^N)^j \nonumber \\
& + \frac{\ell^{2\alpha-1}}{N^{(2\alpha-1)\gamma}}  \lambda_j z^{j,N} (S^N\tS^N x^N)^j
+ \frac{\ell^{3\alpha-1}}{N^{(3\alpha-1)\gamma}}  \lambda_j z^{j,N} (S^N\tSn^2 x^N)^j.\label{tildeQNjcomp}
\end{align}
For the sake of clarity we stress again that in the above $(S^Nx^N)$ is an $N$-dimensional vector and
$(S^Nx^N)^j$ is the $j$-th component of such a vector. Therefore, recalling \eqref{contCS}, \eqref{for1}, \eqref{finitemoments} and setting $\gamma=1/6$, we have
\begin{align}
\sum_{j=1}^N \lambda_j^2 \Epin \lv {Z}^N_j\rv^2 & \less
\frac{1}{N} \sum_{j=1}^N \lambda_j^2 \EE\lv \rho^j z^{j,N}\rv^2
+ \frac{\Epin}{N^{(\alpha-1)/3}} \sum_{j=1}^N \lambda_j^4 \lv (S^Nx^N)^j\rv^2 \nonumber\\
& \quad+ \frac{\Epin}{N^{(2\alpha-1)/3}} \sum_{j=1}^N \lambda_j^4 \lv (S^N\tS^N x^N)^j\rv^2
+ \frac{\Epin}{N^{(3\alpha-1)/3}} \sum_{j=1}^N \lambda_j^4 \lv (S^N\tSn^2 x^N)^j\rv^2  \nonumber \\
& \less\frac{1}{N}+ \frac{1}{N^{(\alpha-1)/3}}\Epin \| \tilde{S}^N x^N\|^2
+ \frac{1}{N^{(2\alpha-1)/3}}\Epin \| \tSn^2 x^N\|^2 \nonumber\\
&\quad
+ \frac{1}{N^{(3\alpha-1)/3}}\Epin \| \tSn^3 x^N\|^2  \nonumber\\
& \less \frac{1}{N}+ \frac{1}{N^{(\alpha-1)/3}}\Epin \|x^N\|^2 \,. \nonumber
\end{align}
Therefore,
\begin{align*}
N^{1/3}\sum_{j=1}^N \lambda_j^2 \Epin \lv {Z}^N_j\rv^2 \less \frac{1}{N^{2/3}}+
\frac{1}{N^{(\alpha-2)/3}}
 \longrightarrow 0,
\quad \mbox{when } \alpha>2.
\end{align*}

\noindent
{\bf Proof of v).} Follows from Assumption \ref{Ass2}, from statement i) of this lemma and from \eqref{tildeeN}.

\noindent
{\bf Proof of vi).} From \eqref{tildeQN},
\begin{align*}
\mathrm{Var}_x({Z}^N) = & \EE_x\left\vert- \frac{\ell^3}{4 {N^{3\gamma}}}
 \lanc{x}{\cC^{1/2} z^N}
 -2 \frac{\ell^{\alpha-1}}{N^{(\alpha-1)\gamma}} \langle \cC^{1/2}z^N, Sx \rangle \right.
\nonumber \\
& \left. - \frac{\ell^{(2\alpha-1)}}{N^{(2\alpha-1)\gamma}} \lanc{\tS \cC_N^{1/2} z^N}{\tS x}
 - \frac{\ell^{(3\alpha-1)}}{N^{(3\alpha-1)\gamma}} \lanc{\tS \cC_N^{1/2} z^N}{\tS^2 x}\right\vert^2 \,.
\end{align*}
Therefore,
\begin{align*}
\mathrm{Var}_x({Z}^N) = & \frac{\ell^6}{16 N^{6\gamma}}\EE_x \norc{x}^2
+ 4 \frac{\ell^{2(\alpha-1)}}{N^{2(\alpha-1)\gamma}} \EE_x \nor{\cC_N^{1/2}S x}^2  \\
& + \frac{\ell^{2(2\alpha-1)}}{N^{2(2\alpha-1)\gamma}}
\EE_x \nor{\cC_N^{1/2}S\tS  x}^2  +
\frac{\ell^{2(3\alpha-1)}}{N^{2(3\alpha-1)\gamma}}
\EE_x \nor{\cC_N^{1/2}S\tS^2  x}^2 + \EE_x r^N
\end{align*}
where $r^N$ contains all the cross-products in the expansion of the variance. By direct calculation and using the antisymmetry of $S$, one finds that most of such cross products vanish and we have
$$
\EE_x r^N:= \frac{\ell^{(2\alpha+2)}}{2 {N^{3\gamma}}N^{(2\alpha-1)\gamma}}
 \langle Sx, \tilde{S} x \rangle + 4
\frac{\ell^{(4\alpha-2)}}{N^{2(2\alpha-1)\gamma}}\EE_x \norc{\tS^2 x}^2.
$$
Observe that
$$
 \langle Sx, \tilde{S} x \rangle = \norc{\tS x}^2;
$$
using this fact,  Assumption \ref{Ass2} implies that the first addend in the above expression for $\EE_x r^N$ vanishes as $N \ra \infty$. The second addend contributes instead to the limiting variance. Now straightforward calculations give the result.

\end{proof}

We recall the definitions
\begin{equation}\label{epsx}
\epsilon\sub^N(x):= \EE_x \left[\left(1 \wedge e^{Q^N} \right) \cC_N^{1/2}z^N \right]
\end{equation}
and
\begin{equation*}
h^N\sub(x):= \EE_x \left(1 \wedge e^{Q^N} \right).
\end{equation*}
\begin{lemma}\label{lemmapreliminaries}
Suppose that  Assumption \ref{AssCLT}, Assumption \ref{Ass2} and Condition \ref{Ass1} hold. Then
\begin{enumerate}
\item If $\alpha > 2$ and $\gamma= 1/6$,
\be\label{est11}
N^{1/3}\Epin \|\epsilon\sub^N(x)\|^2 \stackrel{N\ra \infty}{\longrightarrow}  0\, ;
\ee

\item if $1 \leq \alpha \leq 2$ and $\gamma \geq  1/6$ then
\be\label{est11prime}
\Epin \|N^{\gamma(\alpha-1)}\epsilon\sub^N(x) + 2 \ell^{\alpha-1} \nu\sub\, \tilde{S} x\|^2  \stackrel{N\ra \infty}{\longrightarrow} 0 \,
\ee
where the constant  $\nu\sub$ has been defined in \eqref{bfh}.
\item if $\alpha \geq 1$,  $\gamma \geq  1/6$ and $S^N$ is such that  $(c_1, c_2, c_3) \neq (0,0,0)$, then
\be\label{est12}
\Epin \lv h\sub^N(x)-h\sub\rv^2 \stackrel{N\ra \infty}{\longrightarrow} 0 \,;
\ee
\item finally, if $ 1 < \alpha < 2 $,  $\gamma>1/6$ and $S^N$ is such that $c_1=c_2=c_3=0$, then
\begin{equation*}
\Epin \lv h\sub^N(x)- 1\rv^2 \stackrel{N\ra \infty}{\longrightarrow} 0 \,,
\end{equation*}
i.e.  the constant $h\sub$ in \eqref{est12} is equal to one. This means that, as $N \ra \infty$,  the acceptance probability tends to one.
\end{enumerate}
\end{lemma}
\begin{proof}[Proof of Lemma  \ref{lemmapreliminaries}]
$\bullet \,$ {\bf Proof of (i).} Acting as in \cite[page 2349]{PST}, we obtain
$$
\lv \langle \epsilon^N\sub (x), \varphi_j\rangle \rv^2 \less \lambda_j^2 \EE_x \lv Q_j^N\rv^2 \,.
$$
Taking the sum over $j$ on both sides of the above then gives
$$
\| \epsilon^N\sub (x)\|^2 \less  \sum_{j=1}^N\lambda_j^2 \EE_x \lv Q_j^N\rv^2 \,.
$$
Therefore, if we show
$$
N^{1/3} \sum_{j=1}^N\lambda_j^2 \EE_x \lv Q_j^N\rv^2 \stackrel{N \ra \infty}
{\longrightarrow} 0,
$$
\eqref{est11} follows. From  \eqref{defQj}, it is clear that the above is a consequence of
Lemma \ref{variousestimates} (in particular, it follows from \eqref{lem911}-\eqref{lem914}).

$\bullet \,$ {\bf Proof of (ii).}  Let us split $Q^N$ as follows:
$$
Q^N= R^N+ e^N+i_2^N,
$$
where $ e^N$ and $i_2^N$ are defined in \eqref{defeNN} and \eqref{defi2}, respectively, while
\be\label{defRNfora=2}
R^N:= I^N+i_1+B^N+H^N,
\ee
having set
\begin{align}
I^N  & :=   - \frac{\ell^6}{32 \,N^{6\gamma}} \norc{x^N}^2  -2 \frac{\ell^{2 (\alpha-1)}}{N^{2\gamma (\alpha-1)}}\norc{\tS x^N}^2
- \frac{\ell^{2 (2\alpha -1)}}{2 N^{2\gamma (2\alpha -1)}}\norc{\tS^2 x^N}^2   \label{capitalIN}\\
B^N  &:=      -2 \frac{\ell^{\alpha-1}}{N^{\gamma(\alpha-1)}} \langle \cC^{1/2}z^N, Sx^N \rangle
\label{defBN}\\
H^N &:=     - \frac{\ell^3}{4 {N^{3\gamma}}}
 \lanc{x}{\cC^{1/2} z^N}
- \frac{\ell^{(2\alpha-1)}}{ N^{\gamma(2\alpha-1)}} \lanc{\tS \cC_N^{1/2} z^N}{\tS x^N} \nonumber\\
& - \frac{\ell^{(3\alpha-1)}}{N^{\gamma(3\alpha-1)}} \lanc{\tS \cC_N^{1/2} z^N}{\tS^2 x^N} \,,
\label{defHN}
\end{align}
and $i_1^N$ is defind in \eqref{defi1}. The $j$-th component of $N^{\gamma(\alpha-1)} \epsilon\sub^N$ can be therefore expressed as follows:
\begin{align}\label{remT0}
N^{\gamma(\alpha-1)} \epsilon\sub^{j,N}=N^{\gamma(\alpha-1)}   \EE_x\left[ (1 \wedge e^{Q^N}) \lambda_j z^{j,N}\right] =
N^{\gamma(\alpha-1)}   \EE_x\left[ (1 \wedge e^{R^N}) \lambda_j z^{j,N}\right] + T_0^j
\end{align}
where $T_0^j:=\langle T_0, \varphi_j \rangle $ and
\be\label{defT0}
T_0:= N^{\gamma(\alpha-1)}   \EE_x\left[ \left((1 \wedge e^{Q^N}) - (1 \wedge e^{R^N})\right) \cC_N^{1/2} z^N  \right].
\ee
We now decompose $R^N$ into a component which depends on $z^{j,N}$, $R^N_j$, and a component that does not depend on $z^{j,N}$, $R^N_{j, \perp}$:
$$
R^N=R^N_j+ R^N_{j, \perp},
$$
with
$$
R^N_j:=(i_1)_j+(B^N)_j+(H^N)_j,
$$
having denoted by $ (i_1^N)_j, (B^N)_j$ and $(H^N)_j$  the part of
$i_1, B^N$ and $H^N$, respectively, that depend on $z^{j,N}$. That is,
$$
(i_1)_j=- \frac{\ell^4}{8N^{4 \gamma}} \lv z^{j,N}\rv^2;
$$
as for  $(H^N)_j$, it suffices to notice that
$$
(B^N)_j+(H^N)_j = {Z}^N_j,
$$
and the expression for ${Z}^N_j$ is detailed in \eqref{tildeQNjcomp} (just set $\alpha=2$ in \eqref{tildeQNjcomp}).
With this notation, from \eqref{remT0}, we further write
\be\label{step1}
N^{\gamma(\alpha-1)}   \EE_x\left[ (1 \wedge e^{Q^N}) \lambda_j z^j\right] =
N^{\gamma(\alpha-1)}   \EE_x\left[ (1 \wedge e^{[R^N -(i_1^N)_j)-H_j^N]} \lambda_j z^j\right] + T_0^j+ T_1^j
\ee
where, like before, $T_1^j:=\langle T_1, \varphi_j \rangle $ and
$$
T_1:= N^{\gamma(\alpha-1)} \ell  \EE_x\left[ \left((1 \wedge e^{R^N}) - \left(1 \wedge e^{R^N -(i_1^N)_j-H_j^N}\right)\right) \cC_N^{1/2} z^N  \right].
$$
We recall that the notation $\EE_x$ denotes expected value given $x$, where the expectation is taken over all the sources of noise contained in the integrand. In order to further evaluate the RHS of \eqref{step1} we calculate the expected value of the integrand with respect to the law of $z^j$ (we denote such expected value by $\EE^{z^j}$ and use $\EE^{z^j_-}$ to denote expectation with respect to $z^N\setminus z^{j,N}$); to this end, we  use the following lemma.
\begin{lemma}\label{lemgaussmin3}
If $G$ is a normally distributed random variable with $G \sim \mathcal{N}(0,1)$ then
$$
\EE \left[ G\left( 1 \wedge e^{\delta G + \mu}\right)\right]= \delta e^{\mu + \delta^2/2} \Phi\left( -\frac{\mu}{\lv \delta \rv}- \lv  \delta \rv \right),
$$
where $\Phi$ is the CDF of a standard Gaussian.
\end{lemma}
We apply the above lemma with $\mu= R^N_{j,\perp}$ and $\delta= \delta^B_j$, where
\begin{align*}
\delta^B_j &:= -2 \frac{\ell^{\alpha-1}}{N^{\gamma(\alpha-1)}} \lambda_j  (Sx^N)^j.
\end{align*}
We therefore obtain
\begin{align*}
N^{\gamma(\alpha-1)}  \lambda_j  \EE_x\left[ (1 \wedge e^{[R^N -(i_1^N)_j-H_j^N)]} z^j\right] & =
N^{\gamma(\alpha-1)}  \lambda_j  \EE_{x}^{z^j_-}\delta_j^B e^{R_{j,\perp}^N +(\delta_j^B)^2/2}
\Phi \left( - \frac{R_{j,\perp}^N}{\lv \delta_j^B\rv} - \lv \delta_j^B \rv \right) \nonumber\\
& =-2 \ell^{\alpha-1}  \lambda_j^2 (S^Nx^N)^j \EE_x^{z}
e^{R_{j,\perp}^N +(\delta_j^B)^2/2}
\Phi \left( - \frac{R_{j,\perp}^N}{\lv \delta_j^B\rv} - \lv \delta_j^B \rv \right) \nonumber\\
 &= -2 \ell^{\alpha-1} \lambda_j^2  \EE_{x}(Sx)^j  e^{R_{j,\perp}^N +(\delta_j^B)^2/2}\mathbf{1}_{\{R_{j,\perp}^N <0\}}+ T_2^j+T_3^j\\
&=-2 \ell^{\alpha-1} \lambda_j^2 \EE_{x} (Sx)^j  e^{\tilde{Q}^N }\mathbf{1}_{\{\tilde{Q}<0\}}+ T_2^j+T_3^j +T_4^j\\
&= -2 \ell^{\alpha-1} \lambda_j^2  \EE_{x}(Sx)^j  e^{{Q}^N }\mathbf{1}_{\{{Q}<0\}}+ T_2^j+T_3^j +T_4^j
+ T_5^j,
\end{align*}

\begin{align*}
T_2^j & :=   -2 \ell^{\alpha-1} \lambda_j^2  \EE_{x}(Sx)^j  e^{R_{j,\perp}^N +(\delta_j^B)^2/2} \left[\Phi \left( - \frac{R_{j,\perp}^N}{\lv \delta_j^B\rv} - \lv \delta_j^B \rv \right) - \Phi \left( - \frac{R_{j,\perp}^N}{\lv \delta_j^B\rv}  \right)  \right]
\end{align*}
and
\begin{align}
T_3^j & := -2 \ell^{\alpha-1} \lambda_j^2 \EE_{x} (Sx)^j  e^{R_{j,\perp}^N+(\delta_j^B)^2/2 }  \left[
 \Phi \left( - \frac{R_{j,\perp}^N}{\lv \delta_j^B\rv} \right) -  \mathbf{1}_{\{R_{j,\perp}^N <0\}}   \right]\nonumber
\\
T_4^j & := -2 \ell^{\alpha-1} \lambda_j^2 \EE_{x} (Sx)^j  \left[
e^{R_{j,\perp}^N +(\delta_j^B)^2/2 }  \mathbf{1}_{\{R_{j,\perp}^N <0\}} -  e^{{Q}^N }
\mathbf{1}_{\{{Q}^N<0\}}\right]\nonumber\\
T_5^j & := -2 \ell^{\alpha-1} \lambda_j^2 \EE_{x} (Sx)^j  \left[  e^{{Q}^N }\mathbf{1}_{\{{Q}^N<0\}}
- e^{{Q} }\mathbf{1}_{\{{Q}<0\}}\right]. \label{defT5}
\end{align}

To prove the statement it suffices to show that
$$
\Epin \sum_{n=0}^5\|T_n\|^2\ra 0 \quad \mbox{as } N \ra \infty
$$
These calculations are a a bit lengthy, so we gather the proof of the above in  Lemma \ref{lemmaTis} below. Assuming for the moment that the above is true, the proof is concluded after recognising that
$$-2 \ell^2 \lambda_j^2  (Sx)^j \EE  e^{{Q} }\mathbf{1}_{\{{Q}<0\}}= -2 \ell^2   (\tS^Nx^N)^j \nu\sub.
$$

$\bullet \,$ {\bf Proof of (iii).} By acting as in  the proof of \cite[Lemma 4.5 and Corollary 4.6]{PST} we see that \eqref{est12} is a consequence of \eqref{q=Q^N+estar},  \eqref{ve6} and of the following limit:
$$
\Epin \lv e^N_{\star}\rv^2 \longrightarrow 0  \quad \mbox{as } N \ra \infty.
$$
The above follows from the definition \eqref{estar},  Lemma \ref{variousestimates}, Assumption  \ref{Ass2} and
\cite[equation (4.7)]{PST}.

$\bullet \,$ {\bf Proof of (iv).} One could show this with the same procedure as in (iii). However, in this case things are easier, indeed we can write
\begin{align*}
\Epin \lv \EE_x (1\wedge e^{Q^N}) - (1\wedge e^0)\rv^2
\leq \Epin \lv Q^N\rv^2 \ra 0\,.
\end{align*}
The above limit follows simply by the assumption that $\gamma>1/6$ and $c_1=c_2=c_3=0$.
\end{proof}

\begin{lemma}\label{lemmaTis}
 If Assumption \ref{AssCLT}, Assumption \ref{Ass2},  Assumption \ref{extrassT2} and Condition \ref{Ass1} hold,  then
\begin{equation*}
\Epin \|T_n\|^2= \Epin \sum_{i=1}^N \lv T_n^j \rv^2 \stackrel{N \ra \infty}{\longrightarrow} 0, \quad
\mbox{for all } n \in \{0,1 \dd 5\},
\end{equation*}
where $T_n=\sum_{i=1}^N \langle T_n, \varphi_i\rangle \varphi_i$ and  the terms $T_0^j \dd T_5^j$ have been introduced in  \eqref{defT0}- \eqref{defT5}.
\end{lemma}
\begin{proof}[Proof of Lemma \ref{lemmaTis}] We consecutively prove the above limit for the terms $T_0^j \dd T_5^j$.

$\bullet$ Using the Lipshitzianity of the function $u \ra 1 \wedge e^u$, the result for $T_0$ follows from  the definition of $R^N$, equation \eqref{defRNfora=2}, and Lemma \ref{variousestimates}, statements iii) and v).  The result for $T_1$ can be obtained similarly.

$\bullet$ Term $T_2$: we using the the lipshitzianity of the function $\Phi$ and observe that the following holds
$$
\Epin e^{c (R^N_{j,\perp}+ \delta^2/2)} \less \Epin e^{c R^N}\less 1 \quad \mbox{for all } c>0.
$$
The above can be obtained with a reasoning similar to the one detailed in \cite[page 916 and (5.20)]{{MattinglyPillaiStuart2011}}, using (i) of Assumption \ref{extrassT2} . Using the above observations and applying the Hoelder inequality with the exponent $r$ appearing in (ii) of Assumption \ref{extrassT2}, one then gets,
$$
\Epin \|T_2\|^2 \less \left(\Epin \left( \sum_{j=1}^N \frac{\lambda_j^6 \lv (Sx)^j \rv^4}{N^{2\gamma(\alpha-1)}} \right)^r \right)^{1/r} \,.
$$

Therefore the term $T_2$ goes to zero by Assumption \ref{extrassT2}.

$\bullet$ The term $T_3$ can be treated with calculations completely analogous to those in \cite[Lemma 5.8]{MattinglyPillaiStuart2011}. As a result of such calculations,  using the fact that the noise $z^{j,N}$ is always independent on the current position $x$, and recalling equation \eqref{tildeQNjcomp},   we obtain that for any $r,q>1$ (to be later appropriately chosen), the following bound holds:
\begin{align} \label{tt33}
\Epin \| T_3\|^2 & \less \left\{ \Epin \left[  \sum_{j=1}^N \lambda_j^4 \lv (Sx)^j\rv^2
\left[ \EE_x \frac{(\lambda_j \lv (Sx)^j\rv +1 )}{\lv R^N \rv N^{\gamma(\alpha-1)}+1}
\cdot \left( 1+ \frac{\lv z^{j,N}\rv^2}{N^{4\gamma}}+ Z_j^N\right)
\right]^{2/q}
 \right]^r \right\}^{1/r}
\end{align}
Now set
$$
D_N:=  \left( \EE_x \frac{1}{(\lv R^N \rv N^{\gamma(\alpha-1)}+1)^2}
 \right)^{r/q},
$$
so that
\begin{align*}
 \Epin \| T_3\|^2 & \less \left\{ \Epin D_N  \left\{\sum_{j=1}^N \lambda_j^4 \lv (Sx)^j\rv^2\left[   \EE_x \left( \lambda_j \lv (Sx)^j\rv  +1 \right)^{2}
 \right]^{1/q}
  \right\}^r\right\}^{1/r} \quad\quad\quad\quad\hfill{(I)}\\
& + \left\{ \Epin D_N
\left\{\sum_{j=1}^N \lambda_j^4 \lv (Sx)^j\rv^2\left[   \EE_x \left( \lambda_j \lv (Sx)^j\rv  +1 \right)^{2} \lv Z^N_j\rv^2
 \right]^{1/q} \right\}^r
\right\}^{1/r} \,.  \quad\hfill{(II)}
\end{align*}
Notice that by the bounded convergence theorem, we have
\be\label{dnt0}
\Epin D_N \less \left( \frac{1}{N^{\gamma(\alpha-1)}}\right)^{r/q}.
\ee
With this observation it is easy to show that the term ($I$) tends to zero. It is less easy to show that ($II$) tends to zero, so for this term we detail calculations a bit more.
\begin{align*}
(II) & \less (\Epin D_N^2)^{1/(2r)}
\left( \Epin \left\{ \sum_{j=1}^N \lambda_j^4 \lv (Sx)^j\rv^2
\left[  \EE_x  \left[ (\lambda_j \lv (Sx)^j\rv  +1 ) \frac{\lv x^{j,N} z^{j,N}\rv}{N^{3\gamma}\lambda_j} \right]^2 \right]^{1/q}
    \right\}^{2r}   \right)^{1/(2r)}\\
& + (\Epin D_N^2)^{1/(2r)}
\left( \Epin \left\{ \sum_{j=1}^N \lambda_j^4 \lv (Sx)^j\rv^2
\left[  \EE_x  \left[ (\lambda_j \lv (Sx)^j\rv  +1 ) \frac{\lambda_j \lv (Sx)^j\rv \lv z^{j,N} \rv}{N^{\gamma (\alpha-1)}} \right]^2 \right]^{1/q}
    \right\}^{2r}   \right)^{1/(2r)}\\
&+ (\Epin D_N^2)^{1/(2r)}
\left( \Epin \left\{ \sum_{j=1}^N \lambda_j^4 \lv (Sx)^j\rv^2
\left[  \EE_x  \left[ (\lambda_j \lv (Sx)^j\rv  +1 ) \frac{\lambda_j \lv (S\tS x)^j\rv \lv z^{j,N} \rv}{N^{\gamma (2\alpha-1)}} \right]^2 \right]^{1/q}
    \right\}^{2r}   \right)^{1/(2r)}\\
& + (\Epin D_N^2)^{1/(2r)}
\left( \Epin \left\{ \sum_{j=1}^N \lambda_j^4 \lv (Sx)^j\rv^2
\left[  \EE_x  \left[ (\lambda_j \lv (Sx)^j\rv  +1 ) \frac{\lambda_j \lv (S \tS^2 x)^j\rv \lv z^{j,N} \rv}{N^{\gamma (3\alpha-1)}} \right]^2 \right]^{1/q}
    \right\}^{2r}   \right)^{1/(2r)}\,.
\end{align*}
We denote by ($II$)$_1$ to ($II$)$_4$ the terms in line 1 to 4 of the above array of equations and the scond factor in  line $i$ we denote by ($II$)$_{ib}$, so e.g.
$$
(II)_1= (\Epin D_N^2)^{1/(2r)} ((II)_{1b})^{1/(2r)},
$$
where
$$
(II)_{1b}:=  \Epin \left\{ \sum_{j=1}^N \lambda_j^4 \lv (Sx)^j\rv^2
\left[  \EE_x  \left[ (\lambda_j \lv (Sx)^j\rv  +1 ) \frac{\lv x^{j,N} z^{j,N}\rv}{N^{3\gamma}\lambda_j} \right]^2 \right]^{1/q}
    \right\}^{2r} .
$$
To streamline the presentation we have written the calculations leading to the above four addends in a way that it looks like the choice of $q$ should be the same for the four terms above. However, acting appropriately in the computations that give \eqref{tt33}, one can see that the $q$ does not need to be the same for each one of the above addends.
We show how to study ($II$)$_1$ and ($II$)$_3$, the other terms can be done with similar tricks. Starting from ($II$)$_1$, because of \eqref{dnt0}, we just need to prove that ($II$)$_{1b}$ is bounded. We will do slightly better in what follows. Recall that by assumption
\be\label{allmoments}
\Epin \sum_{j=1}^N \frac{\lambda_j^{2p} \lv (Sx)^j\rv^{2p}}{N^{2p\gamma(\alpha-1)}}\leq \Epin  \left(
\sum_{j=1}^N \frac{\lambda_j^{2} \lv (Sx)^j\rv^{2}}{N^{2\gamma(\alpha-1)}}  \right)^p < \infty.
\ee
Choosing $q=2$ in the definition of  ($II$)$_{1b}$ and recalling $x^{j,N} \sim \lambda_j \rho_j$,  we get
$$
(II)_{1b} = \Epin  \left(
\sum_{j=1}^N \frac{\lambda_j^{2} \lambda_j^3  \lv (Sx)^j\rv^{3}}{N^{3\gamma}}  \right)^{2r} \less \Epin \sum_{j=1}^N \frac{\lambda_j^{2} \lambda_j^{6r}  \lv (Sx)^j\rv^{6r}}{N^{6r\gamma}}   \longrightarrow 0,
$$
where in the last inequality we have used the weighted Jentsen's inequality (relying on the fact that $\{\lambda_j^2\}_j$ is summable) and the convergence of the RHS to zero follows from \eqref{allmoments}. The term ($II$)$_{2b}$ can be dealt with analogously, choosing $q=4$ (this time when applying the weighted Jentsen's inequality one should rely on the fact that the sequence $\{\lambda_j^4 \lv (Sx)^j\rv^2  \}_j $ is summable for every $x \in \cH$). Finally, to deal with ($II$)$_{3b}$, we use the fact that the sequence $\{(\tS^2 x)^j\}_j$ is, by assumpion, bounded for every $x \in \cH$. Therefore, choosing $q=2$ we have:
\begin{align*}
(II)_{3b} &  = \Epin \left( \sum_{j=1}^N
\lambda_j^4  \lv (Sx)^j\rv^{2}  \lambda_j  \lv (Sx)^j\rv \frac{\lambda_j  \lv (S \tS x)^j\rv}{N^{\gamma(2\alpha -1)}}
\right)^{2r}\\
&= \Epin \left( \sum_{j=1}^N
\frac{\lambda_j^4  \lv (Sx)^j\rv^3 }{N^{\gamma (2\alpha-1)}}  \lv (\tS^2 x)^j\rv
\right)^{2r}\\
& \leq  \Epin \left( \sum_{j=1}^N
\frac{\lambda_j^2 \lv (Sx)^j\rv^2}{N^{\gamma (2\alpha-1)}} \lambda_j^2\lv (Sx)^j\rv
\right)^{2r} \,.
\end{align*}
Because  $2\alpha \gamma -\gamma > 2\gamma (\alpha-1)$, the RHS of the above tends to zero by using \eqref{allmoments}. The term ($II$)$_{4b}$ can be dealt with in a completely analogous manner.

$\bullet$ The terms $T_4$ and $T_5$ can be studied similarly to what has been done in \cite{KOS}, see calculations from equation (8.31), in particular the terms $e^{i,N}_{3,k},e^{i,N}_{5,k}$.
\end{proof}


\bibliographystyle{plain}

\end{document}